\documentclass{amsart}
\long\def\sidebyside#1#2{%
 \hbox to\textwidth{\vtop{\hsize=.6\textwidth%

 \advance\hsize by -.32\columnsep
\parindent=0pt
\centering

 #1\vskip1sp}\hskip\columnsep\vtop{\hsize=.6\textwidth%
 \advance\hsize by -.32\columnsep
\parindent=0pt
\centering
#2

}\hfill}}

\usepackage{amsthm, amssymb, amsfonts}
\usepackage{psfrag}
\usepackage{graphicx}
\normalsize
\theoremstyle{definition}
\newtheorem{Thm}{Theorem}

\newtheorem{definition}[Thm]{Definition}
\newtheorem{theorem}[Thm]{Theorem}
\newtheorem{remark}[Thm]{Theorem}

\newcommand {\Ab}{\mathbf{A}}  
\newcommand {\Bb}{\mathbf{B}}
\newcommand {\Cb}{\mathbb{C}}
\newcommand {\Cbf}{\mathbf{C}}
\newcommand {\Db}{\mathbb{D}}

\newcommand {\op}{\mathbf{\oplus}} 
\newcommand {\om}{\mathbf{\ominus}} 
\newcommand {\od}{\mathbf{\otimes}}   
\newcommand {\sqp}{\boxplus}
\newcommand {\sqm}{\boxminus}

\newcommand {\asubM}{\|\ab\|_{\lower.1ex \hbox {\scriptsize {M}}}}
\newcommand {\hsubM}{\|\hb\|_{\lower.1ex \hbox {\scriptsize {M}}}}
\newcommand {\hsubMs}{\|\hb\|_{\lower.1ex \hbox {\scriptsize {M}}}^2}
\newcommand {\hasubM}{\|\hb_{\ab}\|_{\lower.1ex \hbox {\scriptsize {M}}}}
\newcommand {\hbsubM}{\|\hb_{\bb}\|_{\lower.1ex \hbox {\scriptsize {M}}}}
\newcommand {\hcsubM}{\|\hb_{\cb}\|_{\lower.1ex \hbox {\scriptsize {M}}}}
\newcommand {\AsubM}{\|\Ab\|_{\lower.1ex \hbox {\scriptsize {M}}}}
\newcommand {\AsubMs}{\|\Ab\|_{\lower.1ex \hbox {\scriptsize {M}}}^2}
\newcommand {\AosubM}{\|A_1\|_{\lower.1ex \hbox {\scriptsize {M}}}}
\newcommand {\AtsubM}{\|A_2\|_{\lower.1ex \hbox {\scriptsize {M}}}}
\newcommand {\BsubM}{\|\Bb\|_{\lower.1ex \hbox {\scriptsize {M}}}}
\newcommand {\CsubM}{\|\Cb\|_{\lower.1ex \hbox {\scriptsize {M}}}}
\newcommand {\CsubMf}{\|\Cbf\|_{\lower.1ex \hbox {\scriptsize {M}}}}
\newcommand {\CsubMs}{\|\Cb\|_{\lower.1ex \hbox {\scriptsize {M}}}^2}
\newcommand {\asupM}{\|\ab\|^{\lower.1ex \hbox {\scriptsize {M}}}}
\newcommand {\aspM}{\ab^{\lower.1ex \hbox {\scriptsize {M}}}}
\newcommand {\asbM}{\ab_{\lower.1ex \hbox {\scriptsize {M}}}}
\newcommand {\AsupM}{\|\Ab\|^{\lower.1ex \hbox {\scriptsize {M}}}}
\newcommand {\BsupM}{\|\Bb\|^{\lower.1ex \hbox {\scriptsize {M}}}}
\newcommand {\CsupM}{\|\Cb\|^{\lower.0ex \hbox {\scriptsize {M}}}}
\newcommand {\CsupMf}{\|\Cbf\|^{\lower.0ex \hbox {\scriptsize {M}}}}
\newcommand {\asubP}{\|\ab\|_{\lower.1ex \hbox {\scriptsize {P}}}}
\newcommand {\asbP}{\ab_{\lower.1ex \hbox {\scriptsize {P}}}}
\newcommand {\AsubP}{\|\Ab\|_{\lower.1ex \hbox {\scriptsize {P}}}}
\newcommand {\BsubP}{\|\Bb\|_{\lower.1ex \hbox {\scriptsize {P}}}}
\newcommand {\CsubP}{\|\Cb\|_{\lower.1ex \hbox {\scriptsize {P}}}}

\newcommand {\phue}{\lower.3ex \hbox {\scriptsize {UE}}}
\newcommand {\pheu}{\lower.3ex \hbox {\scriptsize {EU}}}
\newcommand {\phum}{\lower.3ex \hbox {\scriptsize {UM}}}
\newcommand {\phmu}{\lower.3ex \hbox {\scriptsize {MU}}}
\newcommand {\phme}{\lower.3ex \hbox {\scriptsize {ME}}}
\newcommand {\phem}{\lower.3ex \hbox {\scriptsize {EM}}}

\newcommand {\lowerkluma}{\lower1.5ex \hbox {\phantom{K}}}
\newcommand {\lowerklumb}{\lower3.5ex \hbox {\phantom{K}}}
\newcommand {\lowerklumc}{\lower7.0ex \hbox {\phantom{K}}}
\newcommand {\pp}{\lower.6ex \hbox {\footnotesize {$P_{1} P_{2}$}}}

 \newcommand {\lowAA}{\lower.3ex \hbox {\scriptsize {$\Ab$}}}
 \newcommand {\lowBB}{\lower.3ex \hbox {\scriptsize {$\Bb$}}}
 \newcommand {\lowCC}{\lower.3ex \hbox {\scriptsize {$\Cb$}}}
 \newcommand {\lowA}{\lower.3ex \hbox {\scriptsize {$A$}}}
 \newcommand {\lowB}{\lower.3ex \hbox {\scriptsize {$B$}}}
 \newcommand {\lowC}{\lower.3ex \hbox {\scriptsize {$C$}}}
 \newcommand {\lowE}{\lower.3ex \hbox {\tiny       {$\rm E$}}}
 \newcommand {\lowM}{\lower.3ex \hbox {\tiny       {$\rm M$}}}
 \newcommand {\lowtM}{\lower.01ex \hbox {\tiny       {$\rm M$}}}
 \newcommand {\lowtE}{\lower.01ex \hbox {\tiny       {$\rm E$}}}
 \newcommand {\lowtU}{\lower.01ex \hbox {\tiny       {$\rm U$}}}

 \newcommand {\lowEC}{\lower.3ex \hbox {\scriptsize {$EC$}}}
 \newcommand {\lowER}{\lower.3ex \hbox {\scriptsize {$ER$}}}
 \newcommand {\lowU}{\lower.3ex \hbox {\tiny       {\rm U}}}
 \newcommand {\lowDU}{\lower.3ex \hbox {\scriptsize {\rm DU}}}
 \newcommand {\lowCU}{\lower.3ex \hbox {\scriptsize {\rm CU}}}
 \newcommand {\lowDM}{\lower.3ex \hbox {\scriptsize {\rm DM}}}
 \newcommand {\lowCM}{\lower.3ex \hbox {\scriptsize {\rm CM}}}
 \newcommand {\lowo}{\lower.3ex \hbox {\scriptsize {\rm 0}}}
 \newcommand {\lowf}{\lower.3ex \hbox {\scriptsize {\rm f}}}

\newcommand {\lowmbpa}{\lower.6ex \hbox {\footnotesize {$\ome\bb\ope\ab$}}}
 
\newcommand {\uvc}{\displaystyle\frac{\lower.6ex \hbox {$\ub\ccdot\vb$}}{c^2}}
\newcommand {\uvs}{\displaystyle\frac{\lower.6ex \hbox {$\ub\ccdot\vb$}}{s^2}}
\newcommand {\unpuvc}{ \lower.6ex \hbox {$1 + \uvc$} }
\newcommand {\unpuvs}{ \lower.6ex \hbox {$1 + \uvs$} }
\newcommand {\uvcbar}{\displaystyle\frac{\lower.6ex \hbox
            {$\ubar\ccdot\vb$}}{c^2}}
\newcommand {\unpuvcbar}{ \lower.6ex \hbox {$1 + \uvcbar$} }
\newcommand {\vwc}{\displaystyle\frac{\lower.6ex\hbox{$\vb\ccdot\wb$}}{c^2}}
\newcommand {\unpvwc}{ \lower.6ex \hbox {$1 + \vwc$} }
\newcommand {\subE}{\!\lower.1ex \hbox {\tiny E}}
\newcommand {\subG}{\!\lower.1ex \hbox {\tiny G}}
\newcommand {\subH}{\!\lower.1ex \hbox {\tiny H}}
\newcommand {\subEs}{\!\lower.1ex \hbox {\tiny {E,S}}}
\newcommand {\subEt}{\!\lower.1ex \hbox {\tiny {E,2}}}
\newcommand {\subEC}{\!\lower.1ex \hbox {\tiny EC}}
\newcommand {\subU}{\!\lower.1ex \hbox {\tiny U}}
\newcommand {\subM}{\!\lower.1ex \hbox {\tiny M}}
\newcommand {\subC}{\!\lower.1ex \hbox {\tiny C}}
\newcommand {\subbE}{\!\lower.01ex \hbox {\tiny E}}
\newcommand {\subbU}{\!\lower.01ex \hbox {\tiny U}}
\newcommand {\subbC}{\!\lower.01ex \hbox {\tiny C}}
\newcommand {\subbM}{\!\lower.01ex \hbox {\tiny M}}
\newcommand {\subbG}{\!\lower.01ex \hbox {\tiny G}}
\newcommand {\subbH}{\!\lower.01ex \hbox {\tiny H}}
\newcommand {\ope}{\op_{_{\subE}}\!\,}

\newcommand {\ome}{\om_{_{\subE}}\!\,} 
\newcommand {\ode}{\od_{_{\subE}}\!\,} 
\newcommand {\sqpe}{\sqp_{_{\,\subbE}}\!\,}

\newcommand {\opm}{\op_{_{\subM}}\!\,} 
\newcommand {\omm}{\om_{_{\subM}}\!\,} 
\newcommand {\odm}{\od_{_{\subM}}\!\,} 
\newcommand {\sqpm}{\sqp_{_{\,\subbM}}\!\,} 
 
\newcommand {\ccdot}{\mathbf{\cdot }} 

\newcommand {\ab}{\mathbf{a}}
\newcommand {\bb}{\mathbf{b}}
\newcommand {\cb}{\mathbf{c}}
\newcommand {\db}{\mathbf{d}}

\newcommand {\hb}{\mathbf{h}}

\newcommand {\ub}{\mathbf{u}}
\newcommand {\vb}{\mathbf{v}}

\newcommand {\wb}{\mathbf{w}}
\newcommand {\bz}{\mathbf{0}}

\newcommand {\xb}{\mathbf{x}}

\newcommand {\zerb}{\mathbf{0}}

\newcommand {\ubar}{\bar{\ub}}

\newcommand {\CC}{\mathbb{C}}

\newcommand {\DD}{\mathbb{D}}
\newcommand {\NN}{\mathbb{N}}

\newcommand {\Rb}{\mathbb{R}}

\newcommand {\Rn}{\Rb^n}

\newcommand {\Rctwou}{{\Rb}_{c=1}^2}

\newcommand {\Rstwo}{{\Rb}_{s}^2}

\newcommand {\Rcn}{{\Rb}_{c}^{n}}

\newcommand {\Rsn}{{\Rb}_{s}^{n}}

\newcommand {\Rct}{{\Rb}_{c}^{3}}

\newcommand {\Rctwo}{{\Rb}_{c}^{2}}
\newcommand {\Rt}{\Rb^3}

\newcommand {\Rtwo}{\Rb^2}

\newcommand {\gub}{\gamma_{\ub}^{\phantom{1}}}
\newcommand {\gvb}{\gamma_{\vb}^{\phantom{1}}}

\newcommand {\gvbi}{\gamma_{\vb_i}^{\phantom{O}}}

\newcommand {\gvbie}{\gamma_{\vb_{i,e}}^{\phantom{O}}}

\newcommand {\gwb}{\gamma_{\wb}^{\phantom{1}}}

\newcommand {\gubs}{\gamma_{\ub}^2}
\newcommand {\gvbs}{\gamma_{\vb}^2}

\newcommand {\gupvb}{\gamma_{\ub\op\vb}^{\phantom{1}}}

\newcommand {\gupvbm}{\gamma_{\ub\opm\vb}^{\phantom{1}}}

\newcommand {\gyr}{{\rm gyr}}

\newcommand {\Aut}{{\rm Aut}}

\newcommand {\gyrab}{\gyr[a,b]}

\newcommand {\gammaab}{\gamma_{12}^{\phantom{O}}}

\newcommand {\gyrba}{\gyr[b,a]}

\newcommand {\gyruvb}{\gyr[\ub,\vb]}

\newcommand {\gyrvub}{\gyr[\vb,\ub]}
\newcommand {\vi}{\mathbb{V}}

\newcommand {\vc}{\mathbb{V}_{\!c}}

\newcommand {\vs}{\mathbb{V}_s}

\newcommand {\timess}{\!\times\!}

\newcommand {\ro}{r_{_1}}
\newcommand {\rt}{r_{_2}}

\newcommand {\half}{\textstyle\frac{1}{2}}

\newcommand {\son}{\textstyle{SO}(n)}

\newcommand {\inn}{\hspace{-0.1cm}\in\hspace{-0.1cm}}

 \newcommand {\LAB}{L_{^{AB}}^{\phantom{o}}}

\newcommand {\subEA}{\!\lower.1ex \hbox {\tiny EA}}

 \newcommand {\MAB}{M_{^{AB}}^{\phantom{o}}}
 \newcommand {\MAD}{M_{^{AD}}^{\phantom{o}}}
 \newcommand {\MBC}{M_{^{BC}}^{\phantom{o}}}
 \newcommand {\PAB}{P_{^{AB}}^{\phantom{o}}}
 \newcommand {\MABDC}{M_{^{ABDC}}^{\phantom{o}}}
 
 \newcommand {\gmA}{\gamma_{_A}}
 \newcommand {\gmB}{\gamma_{_B}}
 \newcommand {\gmC}{\gamma_{_C}}
 \newcommand {\gmD}{\gamma_{_D}}

\begin{document}
\begin{center}
\huge{
M\"obiubs Transformation and Einsten Velocity Addition
in the Hyperbolic Geometry of Bolyai and Lobachevsky
     }
\end{center}
\begin{center}
Abraham A. Ungar\\
Department of Mathematics\\
North Dakota State University\\
Fargo, ND 58105, USA\\
Email: abraham.ungar@ndsu.edu\\
\end{center}

\begin{quotation}
{\bf Abstract}
In this chapter,
dedicated to the 60th  Anniversary of Themistocles M.~Rassias,
M\"obius transformation and Einstein velocity addition meet
in the hyperbolic geometry of Bolyai and Lobachevsky.
It turns out that  M\"obius addition that is extracted from M\"obius transformation
of the complex disc
and Einstein addition from his special theory of relativity
are isomorphic in the sense of gyrovector spaces.
\end{quotation}

\section{Introduction} \label{secc1}

Einstein addition law of relativistically admissible velocities is isomorphic
to M\"obius addition that is extracted from the common M\"obius transformation
of the complex open unit disc. Accordingly, both Einstein addition and
M\"obius addition in the open unit ball of the
Euclidean $n$-space possess the structure of a gyrovector space that
forms a natural powerful generalization of the common vector space structure.
Einstein and M\"obius gyrovector spaces
continue to attract research interest as novel algebraic settings for
hyperbolic geometry, giving rise to the incorporation of
Cartesian coordinates and vector algebra into the study of
the hyperbolic geometry of Bolyai and Lobachevsky \cite{mybook03,mybook05}.
Outstanding novel results and elegant compatibility with well-known results in
hyperbolic geometry make the novel
gyrovector space approach to analytic hyperbolic geometry \cite{mybook04}
an obvious contender for augmenting the traditional way of studying hyperbolic geometry
synthetically.

Professor Themistocles M.~Rassias'
special predilection and contribution to the study of M\"obius transformations
is revealed in his work in the areas of
M\"obius transformations, including
\cite{harukirassias94,harukirassias96,harukirassias98,harukirassias00}
and
\cite{quasi91,service,ahlfors}, along with essential
mathematical developments found, for instance, in
\cite{
rassiaslehto98,
rassiascraio01,
rassiaspras03,
rassiascarat91,
rassiaspras92,
rassiaspras97,
rassiassrivas93}. The latter
contain essential research on geometric
transformations including M\"obius transformations.

The initial purpose of this article,
dedicated to the 60th  Anniversary of Themistocles Rassias,
is to extract M\"obius addition in the ball $\Rcn$
of the Euclidean $n$-space $\Rn$, $n\in\NN$,
from the M\"obius transformation
of the complex open unit disc, and to
demonstrate the hyperbolic geometric isomorphism between the resulting
M\"obius addition and the famous Einstein velocity addition of special relativity
theory.
We will then see that
\begin{enumerate}
\item\label{kamot1}
M\"obius addition in the ball $\Rcn$ forms the algebraic setting for the
Cartesian-Poincar\'e ball model of hyperbolic geometry,
and
\item\label{kamot2}
Einstein addition in the ball $\Rcn$ forms the algebraic setting for the
Cartesian-Beltrami-Klein ball model of hyperbolic geometry,
just as the common
\item\label{kamot3}
vector addition in the space $\Rn$ forms the algebraic setting for the standard
Cartesian model of Euclidean geometry.
\end{enumerate}
Remarkably, Items (1)--(3) enable M\"obius addition in $\Rcn$,
Einstein addition in $\Rcn$, and the standard vector addition in $\Rn$
to be studied comparatively, as in \cite{mybook06}.

Counterintuitively, Einstein velocity addition law of relativistically admissible
velocities is neither commutative nor associative.
The breakdown of commutativity in Einstein addition seemed undesirable to
\'Emile Borel in 1909.
According to the historian of relativity physics Scott Walter,
\cite[Sec.~10]{walter99b}, the famous mathematician
and a former doctoral student of Poincar\'e,
\'Emile Borel (1871-1956) was renowned for his work on the theory of functions,
in which a chair was created for him at the Sorbonne in 1909.
In the years following his appointment he took up the study of relativity theory.
Borel ``fixed'' the seemingly ``defective'' result that Einstein velocity addition
law is noncommutative.
According to Walter, Borel's version of
commutativized relativistic velocity addition involves
a significant modification of Einstein's relativistic velocity composition law.

Contrasting Borel, in this article we commutativize the
Einstein velocity addition law by composing Einstein addition with an appropriate
Thomas precession in a natural way suggested by analogies with the
classical parallelogram addition law
and supported experimentally by cosmological observations of stellar aberration.

Historically, the link between Einstein's special theory of relativity and
the non-Euclidean style was developed during the period 1908-1912 by
Vari{\v c}ak, Robb, Wilson and Lewis, and Borel \cite{walter99b}.
The subsequent development that followed 1912 appeared about 80 years later,
in 2001, as the renowned historian Scott Walter describes
in \cite{walterrev2002}:

\begin{quotation}
Over the years, there have been a handful of attempts to promote the
non-Euclidean style for use in problem solving in relativity and
electrodynamics, the failure of which to attract any substantial
following, compounded by the absence of any positive results must give
pause to anyone considering a similar undertaking.
Until recently, no one was in a position to offer an improvement on the
tools available since 1912. In his [2001] book,
Ungar furnishes the crucial missing
element from the panoply of the non-Euclidean style:
an elegant nonassociative algebraic formalism that fully exploits the
structure of Einstein's law of velocity composition.
The formalism relies on what the author calls the ``missing link''
between Einstein's velocity addition formula and ordinary vector addition:
Thomas precession \ldots
\begin{flushright}
\vspace{-0.4cm}
Scott Walter, 2002 \cite{walterrev2002}
\end{flushright}
\end{quotation}

Indeed, the special relativistic effect known as {\it Thomas precession} is
mathematically abstracted into an operator called a {\it gyrator}, denoted ``gyr''.
The latter, in turn, justifies the prefix ``gyro'' that we extensively use
in {\it gyrolanguage}, where we prefix a gyro to any term that describes a
concept in Euclidean geometry and in associative algebra
to mean the analogous concept in hyperbolic geometry and in nonassociative algebra.
Thus, for instance,
Einstein's velocity addition is neither commutative nor associative,
but it turns out to be both {\it gyrocommutative} and {\it gyroassociative},
giving rise to the algebraic structures known as
gyrogroups and gyrovector spaces.
Remarkably, the mere introduction of the gyrator turns
Euclidean geometry, the geometry of classical mechanics, into
hyperbolic geometry, the geometry of relativistic mechanics.

The breakdown of commutativity in Einstein velocity addition law seemed
undesirable to the famous mathematician \'Emile Borel.
Borel's resulting attempt to ``repair'' the seemingly ``defective'' Einstein velocity
addition in the years following 1912 is described by Walter in
\cite[p.~117]{walter99b}.
Here, however, we see that there is no need to repair Einstein velocity
addition law for being noncommutative since it suggestively
gives rise to the gyroparallelogram law of gyrovector addition, which
turns out to be commutative.
The compatibility of the gyroparallelogram addition law of Einsteinian velocities
with cosmological observations of stellar aberration is explained in
\cite[Chap.~13]{mybook03} and mentioned in \cite[Sec.~10.2]{mybook05}.
The extension of the gyroparallelogram addition law of $k=2$ summands in $\Rcn$ to a
corresponding $k$-dimensional gyroparallelepiped (gyroparallelotope)
addition law of $k>2$ summands is presented in this article
and, with proof, in \cite[Theorem 10.6]{mybook03}.

\section{M\"obius Addition} \label{sece2}


The most general M\"{o}bius transformation of the complex open unit disc
\begin{equation} \label{eqrdfkj01}
\DD= \{z\in \CC :\, |z|<1 \}
\end{equation}
in the complex plane $\CC$
is given by the polar decomposition
\cite{ahlfors73,krantz90},
\begin{equation} \label{eq011}
z \mapsto e^{i\theta} \frac{a+z}{1+\overline{a}z}=e^{i\theta} (a\opm z)
\end{equation}

M\"{o}bius addition $\opm$ in the disc is extracted from \eqref{eq011}, allowing the
generic M\"{o}bius transformation of the disc to be viewed as a
M\"{o}bius left {\it gyrotranslation}
\begin{equation} \label{eq012}
z\mapsto a\opm z=\frac{a+z}{1+\overline{a}z}
\end{equation}
followed by a rotation.
Here $\theta\inn\Rb$ is a real number, $a,z\in\DD$, $\overline{a}$
is the complex conjugate of $a$, and $\opm$ represents M\"obius addition
in the disc.

M\"{o}bius addition $a\opm z$ and subtraction
$a\omm z=a\opm (-z)$ are found useful in the geometric viewpoint of
complex analysis; see, for instance,
\cite{disc99,unification04},\cite[pp.~52--53, 56--57, 60]{krantz90},
and the Schwarz-Pick Lemma in \cite[Theorem 1.4, p.~64]{goebel84}.
However, prior to the appearance of \cite{mybook01} in 2001
these were not considered `addition' and `subtraction'
since it has gone unnoticed that,
being {\it gyrocommutative} and {\it gyroassociative},
they share analogies with the common vector addition and subtraction,
as we will see in the sequel.

M\"{o}bius addition $\opm$ is neither commutative nor associative.
The breakdown of commutativity in
M\"{o}bius addition is "repaired" by the introduction of
a {\it gyrator}
\begin{equation} \label{eq1edjg01}
\gyr:\, \DD\timess\DD \to \Aut (\DD,\opm)
\end{equation}
that generates gyroautomorphisms according to the equation
\begin{equation} \label{eq014}
\gyrab=\frac{a\opm b}{b\opm a}=\frac{1+a\overline{b}}{1+\overline{a}b}
\in \Aut(\DD,\opm)
\end{equation}
where $\Aut (\DD,\opm)$ is the automorphism group of the M\"obius groupoid
$(\DD,\opm)$.
Here a groupoid is a nonempty set with a binary operation,
and an automorphism of the groupoid $(\DD,\opm)$
is a bijective self-map $f:\DD\rightarrow\DD$ of the set $\DD$
that respects its binary operation $\opm$, that is,
$f(a\opm b)=f(a)\opm f(b)$ for all $a,b\in\DD$.
Being gyrations, the automorphisms $\gyrab$ are also called {\it gyroautomorphisms}.

The inverse of the automorphism $\gyrab$ is clearly $\gyrba$,
\begin{equation} \label{eq1edjg02}
\gyr^{-1}[a,b]=\gyrba
\end{equation}

The gyration definition in \eqref{eq014} suggests the following
{\it gyrocommutative law} of M\"obius addition in the disc,
\begin{equation} \label{eq015}
a\opm b=\gyrab(b\opm a)
\end{equation}
The resulting gyrocommutative law \eqref{eq015}
is not terribly surprising since it is generated by definition,
but we are not finished.

Coincidentally, the gyroautomorphism $\gyrab$ that repairs
in \eqref{eq015} the breakdown of commutativity,
repairs the breakdown of associativity in
$\opm$ as well, giving rise to the following
{\it left and right gyroassociative law} of M\"obius addition
 \begin{equation} \label{eq016}
 \begin{split}
a\opm(b\opm z)&=(a\opm b)\opm\gyrab z \\
(a\opm b)\opm z&=a\opm (b\opm\gyrba z)
 \end{split}
 \end{equation}
for all $a,b,z\in\DD$.
Moreover, M\"obius gyroautomorphisms possess their own rich structure
obeying, for instance, the two elegant identities
 \begin{equation} \label{eq016loop}
 \begin{split}
\gyr[a\opm b,b] &= \gyrab \\
\gyr[a,b\opm a] &= \gyrab \\
 \end{split}
 \end{equation}
called the left and the right loop property.

In order to extend M\"obius addition from the disc to the ball,
we identify complex numbers of the complex plane $\CC$ with
vectors of the Euclidean plane $\Rtwo$ in the usual way,
\begin{equation} \label{eq6100}
\CC \ni u = u_1+iu_2 = (u_1,u_2) = \ub \in \Rtwo
\end{equation}
Then
\begin{equation} \label{eq6101}
\begin{array}{c}
\bar{u}v+u\bar{v} = 2\ub\ccdot\vb
\\[8pt]
|u| = \|\ub\|
\end{array}
\end{equation}
give the inner product and the norm in $\Rtwo$,
so that M\"obius addition in the disc $\DD$ of the complex plane $\CC$
becomes M\"obius addition in the disc
\begin{equation} \label{trksh}
\Rctwou=\{\vb\inn\Rtwo:\|\vb\|<s=1\}
\end{equation}
of the Euclidean plane $\Rtwo$. Indeed,
\begin{equation} \label{eq6102}
\begin{split}
\DD \ni
u\op v &:= \frac{u+v}{1+\bar{u}v}                      \\
&= \frac{(1+u\bar{v})(u+v)} {(1+\bar{u}v)(1+u\bar{v})}\\
&= \frac{(1+\bar{u}v+u\bar{v}+|v|^2)u+(1-|u|^2)v}
        {1+\bar{u}v+u\bar{v}+|u|^2|v|^2}              \\
&= \frac{(1+2\ub\ccdot\vb+\|\vb\|^2)\ub+(1-\|\ub\|^2)\vb}
         {1+2\ub\ccdot\vb+\|\ub\|^2\|\vb\|^2}         \\
&=: \ub\op\vb \in \Rctwou
\end{split}
\end{equation}
for all $u,v\in\Db$ and all $\ub,\vb\in\Rctwou$.
The last equation in \eqref{eq6102} is a vector equation, so that
its restriction to the ball of the Euclidean two-dimensional space
is a mere artifact. Suggestively, we thus arrive at the following
definition of M\"obius addition in the ball of any real inner
product space.

\begin{definition}\label{defmobiusadd} 
{\bf (M\"obius Addition in the Ball).}
{\it
Let $\vi = (\vi,+,\ccdot)$
be a real inner product space with a binary operation $+$ and
a positive definite inner product $\ccdot$
(\cite[p.~21]{marsden74}; following \cite{kowalski77}, also known
as Euclidean space) and let
$\vs$ be the $s$-ball of $\vi$,
\begin{equation} \label{eqsball}
\vs = \{\vb\in\vi : \|\vb\|<s\}
\end{equation}
for any fixed $s>0$.
M\"obius addition $\opm$ is a binary operation in $\vs$ given by
the equation
 \begin{equation} \label{eq696}
\ub\opm\vb
 = \frac{(1+\frac{2}{s^2}\ub\ccdot\vb+\frac{1}{s^2}\|\vb\|^2 )\ub
        +(1-\frac{1}{s^2}\|\ub\|^2)\vb}
        {1+\frac{2}{s^2}\ub\ccdot\vb+\frac{1}{s^4}\|\ub\|^2\|\vb\|^2}
 \end{equation}
where $\ccdot$ and $\|\ccdot\|$ are the inner product and norm that the
ball $\vs$ inherits from its space $\vi$.
}
\end{definition}

In the limit of large $s$, $s\rightarrow\infty$, the ball $\vs$ in
Def.~\ref{defmobiusadd} expands to the
whole of its space $\vi$, and M\"obius addition in $\vs$
reduces to the vector addition, +,
in $\vi$. Accordingly, the right hand side of \eqref{eq696}
is known as a M\"obius translation \cite[p.~129]{ratcliffe94}.
An earlier study of M\"obius translation in several dimensions, using the
notation $-\ub\opm\vb=:T_\ub\vb$, is found in \cite{ahlfors81} and
in \cite{ahlfors84}, where it is attributed to Poincar\'e.
Both Ahlfors \cite{ahlfors81} and Ratcliffe \cite{ratcliffe94}, who
studied the M\"obius translation in several dimensions, did not call
it a M\"obius addition since it has gone unnoticed at the time that
M\"obius translation is regulated by algebraic laws analogous to those that
regulate vector addition.

M\"obius addition $\opm$ in the open unit ball $\vs$
of any real inner product space $\vi$ is thus a most natural
extension of M\"obius addition in the open complex unit disc.
Like the M\"obius disc groupoid $(\DD,\opm)$, the M\"obius ball groupoid $(\vs,\opm)$
turns out to be a gyrocommutative gyrogroup,
defined in Defs.~\ref{defroupx}\,--\,\ref{defgyrocomm} in Sec.~\ref{secc3},
as one can check straightforwardly by computer algebra.
Interestingly, the gyrocommutative law of M\"obius addition was already known
to Ahlfors \cite[Eq.~39]{ahlfors81}.
The accompanied gyroassociative law of M\"obius addition, however,
had gone unnoticed.

M\"obius addition
satisfies the
gamma identity
\begin{equation} \label{eq6gupv72}
\gupvbm = \gub\gvb\sqrt
{1+\frac{2}{s^2}\ub\ccdot\vb+\frac{1}{s^4}\|\ub\|^2\|\vb\|^2}
\end{equation}
for all $\ub,\vb\inn\vs$,
where $\gub$ is the gamma factor
\begin{equation} \label{eqeq6gupv72gs}
\gub = \frac{1}{\sqrt{1-\displaystyle\frac{\|\ub\|^2}{s^2}}}
\end{equation}
in the $s$-ball $\vs$.

The gamma factor appears also in Einstein velocity addition of
relativistically admissible velocities, and it is known
in special relativity theory as the Lorentz gamma factor.
The gamma factor $\gvb$ is real if and only if $\vb\inn\vs$.
Hence, the gamma identity \eqref{eq6gupv72} demonstrates that
$\ub,\vb\inn\vs\Rightarrow\ub\opm\vb\inn\vs$
so that, indeed, M\"obius addition $\opm$ is a binary operation
in the ball $\vs$. 

\section{Einstein Velocity Addition} \label{secc2}

Let $c$ be any positive constant, let $(\Rcn,+,\ccdot)$ be the Euclidean $n$-space,
and let
\begin{equation} \label{eqcball}
\Rcn  = \{\vb\in\Rcn: \|\vb\| < c \}
\end{equation}
be the $c$-ball of all relativistically admissible velocities of material
particles. It is the open ball of radius $c$, centered at the
origin of $\Rn$, consisting of all vectors $\vb$
in $\Rn$ with magnitude $\|\vb\|$ smaller than $c$.

Einstein velocity addition in the $c$-ball of all relativistically admissible velocities
is given by the equation
\cite{fock},
\cite[p.~55]{moller52},
\cite[Eq.~2.9.2]{urbantkebookeng},
\cite{mybook01},
\begin{equation} \label{eq01}
{\ub}\op{\vb}=\frac{1}{\unpuvc}
\left\{ {\ub}+ \frac{1}{\gub}\vb+\frac{1}{c^{2}}\frac{\gamma _{{\ub}}}{%
1+\gamma _{{\ub}}}( {\ub}\ccdot{\vb}) {\ub} \right\}
\end{equation}
satisfying the {\it gamma identity}
\begin{equation} \label{eqgupv00}
\gupvb = \gub\gvb\left(1+\frac{\ub\ccdot\vb}{c^2}\right)
\end{equation}
for all $\ub,\vb\in\Rcn $,
where $\gub$ is the gamma factor \eqref{eqeq6gupv72gs},
\begin{equation} \label{v72gs}
\gub = \frac{1}{\sqrt{1-\displaystyle\frac{\|\ub\|^2}{c^2}}}
\end{equation}
in the $c$-ball $\Rcn $.

In physical applications, $\Rn=\Rt$ is the Euclidean 3-space,
which is the space of all classical, Newtonian velocities, and
$\Rcn=\Rct\subset\Rt$ is the $c$-ball of $\Rt$ of all relativistically
admissible, Einsteinian velocities.
Furthermore, the constant $c$ represents in physical applications the
vacuum speed of light.

Einstein addition \eqref{eq01} of relativistically
admissible velocities was introduced by Einstein in his 1905 paper
\cite{einstein05} \cite[p.~141]{einsteinfive}
that founded the special theory of relativity.
We may note here that the Euclidean 3-vector algebra was not so
widely known in 1905 and, consequently, was not used by Einstein.
Einstein calculated in \cite{einstein05} the behavior
of the velocity components parallel and orthogonal to the relative
velocity between inertial systems, which is as close as one can get
without vectors to the vectorial version \eqref{eq01}.

In full analogy with vector addition and subtraction,
we use the abbreviation
$\ub\om\vb=\ub\op(-\vb)$ for Einstein subtraction, so that,
for instance, $\vb\om\vb = \zerb$,
$\om\vb = \zerb\om\vb=-\vb$ and, in particular,
\begin{equation} \label{eq01a}
\om(\ub\op\vb) = \om\ub\om\vb
\end{equation}
and
\begin{equation} \label{eq01b}
\om\ub\op(\ub\op\vb) = \vb
\end{equation}
for all $\ub,\vb$ in the ball. Identity
\eqref{eq01a} is called the {\it automorphic inverse property},
and Identity
\eqref{eq01b} is called the {\it left cancellation law} of Einstein
addition \cite{mybook02,mybook03,mybook04}.
Einstein addition does not obey the immediate right counterpart of the
left cancellation law \eqref{eq01b} since, in general,
\begin{equation} \label{eq01c}
(\ub\op\vb)\om\vb \ne \ub
\end{equation}
However, this seemingly lack of a right cancellation law will be repaired
in \eqref{eq01bb}, following the emergence of a second
gyrogroup binary operation in Def.~\ref{defdual} below,
which we introduce in order to capture analogies with classical results.

In the Newtonian limit of large $c$, $c\rightarrow\infty$, the ball $\Rcn $
expands to the whole of its space $\Rn$, as we see from \eqref{eqcball},
and Einstein addition $\op$ in $\Rcn $
reduces to the common vector addition $+$ in $\Rn$,
as we see from \eqref{eq01} and \eqref{v72gs}.

Einstein addition is noncommutative.
Indeed, $\|\ub\op\vb\|=\|\vb\op\ub\|$, but, in general,
\begin{equation} \label{eqyt01}
\ub\op\vb\ne\vb\op\ub
\end{equation}
$\ub,\vb\in\Rcn $. Moreover, Einstein addition is also nonassociative
since, in general,
\begin{equation} \label{eqyt02}
(\ub\op\vb)\op\wb\ne\ub\op(\vb\op\wb)
\end{equation}
$\ub,\vb,\wb\in\Rcn $.

It seems that following the breakdown of
commutativity and associativity in Einstein addition some mathematical
regularity has been lost in the transition from
Newton velocity addition in $\Rn$ to
Einstein velocity addition \eqref{eq01} in $\Rcn $. This is, however, not the
case since, as we will see in Sec.~\ref{secc3},
the gyrator comes to the rescue
\cite{rassiasrev2008,rassiasrev2010,mybook01,mybook02,mybook03,mybook04,walterrev2002}.
Indeed, we will find in Sec.~\ref{secc3} that the mere introduction of gyrations endows
the Einstein groupoid $(\Rcn,\op)$ with a grouplike rich structure \cite{grouplike}
that we call a {\it gyrocommutative gyrogroup}.
Furthermore, we will find in Sec.~\ref{secc4} that Einstein gyrogroups admit
scalar multiplication that turns them into Einstein gyrovector spaces.
The latter, in turn, form the algebraic setting for the
Cartesian-Beltrami-Klein ball model of hyperbolic geometry, just as
Euclidean vector spaces $\Rn$ form the algebraic setting for the
standard Cartesian model of Euclidean geometry.

When the nonzero vectors $\ub,\vb\in\Rcn\subset\Rn$ are parallel
in $\Rn$, $\ub \| \vb$, that is, $\ub=\lambda\vb$
for some $0\ne\lambda\in\Rb$, Einstein addition reduces to the Einstein
addition of parallel velocities \cite[p.~50]{whittaker49},
\begin{equation} \label{eq1pfck03}
\ub\op\vb = \frac{\ub+\vb}{1+\frac{1}{c^2}\|\ub\|\|\vb\|}, \qquad
\ub \| \vb
\end{equation}
which was confirmed experimentally by Fizeau's 1851 experiment \cite{miller81}.
Owing to its simplicity, some books on special relativity present
Einstein velocity addition in its restricted form \eqref{eq1pfck03}
rather than its general form \eqref{eq01}.

The restricted Einstein addition \eqref{eq1pfck03}
is both commutative and associative.
Accordingly, the restricted Einstein addition is a group operation,
as Einstein noted in \cite{einstein05}; see \cite[p.~142]{einsteinfive}.
In contrast, Einstein made no remark about group properties of his addition law
of velocities that need not be parallel. Indeed, the general Einstein
addition \eqref{eq01} is not a group operation but, rather, a gyrocommutative
gyrogroup operation, a structure that was discovered more than
80 years later, in 1988 \cite{parametrization}, and is presented in
Defs.~\ref{defroupx}\,--\,\ref{defgyrocomm} in Sec.~\ref{secc3}.

\section{Einstein Gyrogroups and Gyrations}
\label{secc3}

A description of the 3-space rotation, which since 1926 \cite{thomas26} is named
after Thomas,
is found in Silberstein's 1914 book \cite{silberstein14}.
In 1914 Thomas precession did not have a name, and
Silberstein called it in his 1914 book a ``certain space-rotation''
\cite[p.~169]{silberstein14}.
An early study of Thomas precession, made by the famous
mathematician \'Emile Borel in 1913, is described in his 1914 book
\cite{borel14} and, more recently, in \cite{stachel95}.
According to Belloni and Reina \cite{belloni86},
Sommerfeld's
route to Thomas precession dates back to 1909.
However, prior to Thomas' discovery the relativistic peculiar
3-space rotation had a most uncertain physical status
\cite[p.~119]{walter99b}.
The only knowledge Thomas had in 1925 about the peculiar relativistic
gyroscopic precession \cite{jonson07}
came from De Sitter's formula describing the
relativistic corrections for the motion of the moon, found in
Eddington's book \cite{eddington24}, which was just published at that time
\cite[Sec.~1, Chap.~1]{mybook01}.

The physical significance of the peculiar rotation in
special relativity emerged in 1925 when Thomas relativistically re-computed
the precessional frequency of the doublet separation in the fine structure
of the atom, and thus rectified a missing factor of 1/2. This correction has
come to be known as the {\it Thomas half} \cite{chrysos06}.
Thomas' discovery of the relativistic precession of the
electron spin on Christmas 1925 thus led to the understanding of the
significance of the relativistic
effect that became known as {\it Thomas precession}.
Llewellyn Hilleth Thomas died in
Raleigh, NC, on April 20, 1992.
A paper \cite{bloch02} dedicated to the centenary of the birth of
Llewellyn H. Thomas (1902\,--\,1992) describes the
Bloch gyrovector of quantum information and computation.

For any $\ub,\vb\inn\Rcn $, let $\gyruvb : \Rcn \rightarrow\Rcn $ be the self-map of $\Rcn $
given in terms of Einstein addition $\op$, \eqref{eq01},
by the equation \cite{parametrization}
\begin{equation} \label{eq004}
\gyruvb \wb = \om(\ub\op\vb)\op\{\ub\op(\vb\op\wb)\}
\end{equation}
for all $\wb\inn\Rcn $.
The self-map $\gyruvb$ of $\Rcn$, which takes $\wb\inn\Rcn $ into
$\om(\ub\op\vb)\op\{\ub\op(\vb\op\wb)\}\inn\Rcn $,
is the gyration generated by $\ub$ and $\vb$.
Being the mathematical abstraction of the relativistic Thomas precession,
The gyration  has an interpretation in hyperbolic geometry
\cite{vermeer05}
as the negative hyperbolic triangle defect \cite[Theorem 8.55]{mybook03}.

In the Newtonian limit, $c\rightarrow\infty$, Einstein addition $\op$ in $\Rcn$
reduces to the common vector addition + in $\Rn$, which is associative.
Accordingly, in this limit the gyration $\gyruvb$ in \eqref{eq004} reduces to the
identity map of $\Rn$, called the trivial map.
Hence, as expected, Thomas gyrations
$\gyruvb$, $\ub,\vb\inn\Rcn $,
vanish (that is, they become {\it trivial}) in the Newtonian limit.

It follows from the gyration equation \eqref{eq004} that gyrations
measure the extent to which Einstein addition
deviates from associativity, where associativity
corresponds to trivial gyrations.

The gyration equation \eqref{eq004} can be manipulated
(with the help of computer algebra) into the equation
\begin{equation} \label{hdge1ein}
\gyruvb\wb = \wb + \frac{A\ub+B\vb}{D}
\end{equation}
where
\begin{equation} \label{hdgej2ein}
\begin{split}
 A &=-\frac{1}{c^2}\frac{\gubs}{(\gub+1)} (\gvb-1) (\ub\ccdot\wb)
 +
 \frac{1}{c^2}\gub\gvb (\vb\ccdot\wb)
\\[8pt] & \phantom{=} ~+
 \frac{2}{c^4} \frac{\gubs\gvbs}{(\gub+1)(\gvb+1)} (\ub\ccdot\vb) (\vb\ccdot\wb)
\\[8pt]
B &=- \frac{1}{c^2}
\frac{\gvb}{\gvb+1}
\{\gub(\gvb+1)(\ub\ccdot\wb) + (\gub-1)\gvb(\vb\ccdot\wb) \}
\\[8pt]
D &= \gub\gvb(1+ \frac{\ub\ccdot\vb}{c^2}) +1 = \gamma_{\ub\op\vb}^{\phantom{O}} + 1 > 1
\end{split}
\end{equation}
for all $\ub,\vb,\wb\in\Rcn $.

Allowing $\wb\in\Rn\supset\Rcn $ in \eqref{hdge1ein}\,--\,\eqref{hdgej2ein},
that is, extending the domain of $\wb$ from $\Rcn$ to $\Rn$,
gyrations $\gyr[\ub,\vb]$
are expendable to linear maps of $\Rn$ for all $\ub,\vb\in\Rcn $.

In each of the three special cases when
(i) $\ub=\zerb$, or
(ii) $\vb=\zerb$, or
(iii) $\ub$ and $\vb$ are parallel
in $\Rcn\subset\Rn$, $\ub\|\vb$, we have
$A\ub+B\vb=\zerb$ so that $\gyr[\ub,\vb]$ is trivial,
\begin{equation} \label{sprdhein}
\begin{split}
\gyr[\zerb,\vb]\wb &= \wb  \\
\gyr[\ub,\zerb]\wb &= \wb \\
\gyr[\ub,\vb]\wb &= \wb , \hspace{1.2cm} \ub\|\vb
\end{split}
\end{equation}
for all $\ub,\vb\in\Rcn$, and all $\wb\in\Rn$.

It follows from \eqref{hdge1ein} that
\begin{equation} \label{eq1ffmznein}
\gyr[\vb,\ub](\gyr[\ub,\vb]\wb) = \wb
\end{equation}
for all $\ub,\vb\in\Rcn $, $\wb\in\Rn$, so that gyrations are invertible
linear maps of $\Rn$,
the inverse of $\gyr[\ub,\vb]$ being $\gyr[\vb,\ub]$
for all $\ub,\vb\in\Rcn $.

Gyrations keep the inner product of
elements of the ball $\Rcn $ invariant, that is,
\begin{equation} \label{eq005}
\gyruvb\ab\ccdot\gyruvb\bb = \ab\ccdot\bb
\end{equation}
for all $\ab,\bb,\ub,\vb\inn\Rcn $. Hence, $\gyruvb$ is an
{\it isometry} of $\Rcn $,
keeping the norm of elements of the ball $\Rcn $ invariant,
\begin{equation} \label{eq005a}
\|\gyruvb \wb\| = \|\wb\|
\end{equation}
Accordingly, $\gyruvb$ represents a rotation of the ball $\Rcn $ about its origin for
any $\ub,\vb\inn\Rcn $.

The invertible self-map $\gyruvb$ of $\Rcn $ respects Einstein addition in $\Rcn $,
\begin{equation} \label{eq005b}
\gyruvb (\ab \op \bb) = \gyruvb\ab \op \gyruvb\bb
\end{equation}
for all $\ab,\bb,\ub,\vb\inn\Rcn $,
so that $\gyruvb$ is an automorphism of the Einstein groupoid $(\Rcn ,\op)$.
We recall that an automorphism of a groupoid $(\Rcn ,\op)$
is a bijective self-map of the groupoid $\Rcn $
that respects its binary operation, that is, it satisfies \eqref{eq005b}.
Under bijection composition the automorphisms
of a groupoid $(\Rcn ,\op)$
form a group known as the automorphism group, and denoted
$\Aut(\Rcn ,\op)$.
Being special automorphisms, gyrations
$\gyruvb \inn \Aut(\Rcn ,\op)$, $\ub,\vb\inn\Rcn $,
are also called {\it gyroautomorphisms}, $\gyr$ being the gyroautomorphism
generator called the {\it gyrator}.

The gyroautomorphisms $\gyruvb$ regulate Einstein addition in the ball $\Rcn $,
giving rise to the following nonassociative algebraic laws that ``repair''
the breakdown of commutativity and associativity in Einstein addition:
 \begin{alignat}{2}\label{laws00}
 \notag
  \ub\op\vb & \!=\! \gyruvb(\vb\op\ub) &&\hspace{0.8cm}\text{Gyrocommutativity}\\
 \notag
  \ub\op(\vb\op\wb)& \!=\! (\ub\op\vb)\op\gyruvb\wb&&\hspace{0.8cm}
  \text{Left Gyroassociativity} \\
 \notag
  (\ub\op\vb)\op\wb& \!=\! \ub\op(\vb \op\gyrvub\wb) &&\hspace{0.8cm}
  \text{Right Gyroassociativity} \\
 \end{alignat}
for all $\ub,\vb,\wb\inn\Rcn $.
It is clear from the identities in \eqref{laws00} that
the gyroautomorphisms $\gyruvb$ measure of the failure of commutativity and associativity
in Einstein addition.

Owing to the gyrocommutative law in \eqref{laws00},
the gyrator is recognized as the familiar Thomas precession of special relativity theory.
The gyrocommutative law was already known to Silberstein in 1914
\cite{silberstein14} in the following sense.
The Thomas precession generated by $\ub,\vb\inn\Rct$
is the unique rotation
that takes $\vb\op\ub$ into $\ub\op\vb$ about an axis perpendicular to the
plane of $\ub$ and $\vb$ through an angle $< \pi$ in $\Rcn$,
thus giving rise to the gyrocommutative law.
Obviously, Silberstein did not use the terms
``Thomas precession'' and ``gyrocommutative law'' since
these terms have been coined
later, respectively, following Thomas' 1926 paper \cite{thomas26},
and by the author in 1991 \cite{grouplike,gaxioms} following the discovery
of the gyrocommutative and the gyroassociative laws of
Einstein addition in \cite{parametrization}.
Thus, contrasting the discovery before 1914 of what we presently call
the gyrocommutative law of Einstein addition,
the gyroassociative laws of Einstein addition, left and right,
were discovered by the author about 75 years later, in 1988 \cite{parametrization}.

Thomas precession has purely kinematical origin,
as emphasized in \cite{ungarthomas06}, so that the presence of
Thomas precession is not connected with the action of any force.

A most important and useful property of gyrations is the so called
{\it loop property} (left and right),
\begin{equation} \label{loops00}
\begin{array}{cl}
{\gyr}[\ub\op\vb,\vb]=\gyruvb
& \hbox{ \ \ \ \ \ \ Left Loop Property} \\[3pt]
{\gyr}[\ub,\vb\op\ub]=\gyruvb
& \hbox{ \ \ \ \ \ \ Right Loop Property} \\
\end{array}
\end{equation}
for all $\ub,\vb\inn\Rcn $.
The left loop property will prove useful in \eqref{eqtghj04} below
in solving a basic gyrogroup equation.

Identities \eqref{laws00}\,--\,\eqref{loops00} are the basic identities
of the gyroalgebra of Einstein addition. They can be verified straightforwardly
by computer algebra, as explained in \cite[Sec.~8]{mybook01}.

The grouplike groupoid $(\Rcn ,\op)$ that regulates Einstein addition,
$\op$, in the ball $\Rcn $ of the Euclidean $n$-space $\Rn$ is a
{\it gyrocommutative gyrogroup} called an {\it Einstein gyrogroup}.
Einstein gyrogroups and gyrovector spaces are studied in
\cite{mybook01,mybook02,mybook03,mybook04}.
Gyrogroups are not peculiar to Einstein addition \cite{mbtogyp08}.
Rather, they are abound in the theory of groups \cite{tuvalungar01,tuvalungar02,feder03},
loops \cite{issa99},
quasigroup \cite{issa2001,kuznetsov03},
and Lie groups \cite{kasparian04,kikkawa75,kikkawa99}.

Thus, the type of structure arising in the study of Einstein velocity addition (and
M\"obius addition) is of rather frequent occurrence
and hence merits an axiomatic approach.
Taking the key features of Einstein velocity addition law as axioms,
and guided by analogies
with groups, we are led to the following formal definition of gyrogroups.

\begin{definition}\label{defroupx}
{\bf (Gyrogroups).}
{\it
A groupoid is a non-empty set with a binary operation.
A groupoid $(G , \op )$
is a gyrogroup if its binary operation satisfies the following axioms.
In $G$ there is at least one element, $0$, called a left identity, satisfying

\noindent
(G1) \hspace{1.2cm} $0 \op a=a$

\noindent
for all $a \inn G$. There is an element $0 \inn G$ satisfying axiom $(G1)$ such
that for each $a\inn G$ there is an element $\om a\inn G$, called a
left inverse of $a$, satisfying

\noindent
(G2) \hspace{1.2cm} $\om a \op a=0\,.$

\noindent
Moreover, for any $a,b,c\inn G$ there exists a unique element $\gyr[a,b]c \inn G$
such that the binary operation obeys the left gyroassociative law

\noindent
(G3) \hspace{1.2cm} $a\op(b\op c)=(a\op b)\op\gyrab c\,.$

\noindent
The map $\gyr[a,b]:G\to G$ given by $c\mapsto \gyr[a,b]c$
is an automorphism of the groupoid $(G,\op)$, that is,

\noindent
(G4) \hspace{1.2cm} $\gyrab\inn\Aut (G,\op) \,,$

\noindent
and the automorphism $\gyr[a,b]$ of $G$ is called
the gyroautomorphism, or the gyration, of $G$ generated by $a,b \inn G$.
The operator $\gyr : G\times G\rightarrow\Aut (G,\op)$ is called the
gyrator of $G$.
Finally, the gyroautomorphism $\gyr[a,b]$ generated by any $a,b \inn G$
possesses the left loop property

\noindent
(G5) \hspace{1.2cm} $\gyrab=\gyr [a\op b,b] \,.$
}
\end{definition}

The first pair of the gyrogroup axioms are like the group axioms.
The last pair present the gyrator axioms and the middle axiom links the two pairs.

As in group theory, we use the notation
$a \om b = a \op (\om b)$
in gyrogroup theory as well.

In full analogy with groups, some gyrogroups are gyrocommutative according to the
following definition.

\begin{definition}\label{defgyrocomm}
{\bf (Gyrocommutative Gyrogroups).}
{\it
A gyrogroup $(G, \oplus )$ is gyrocommutative if
its binary operation obeys the gyrocommutative law

\noindent
(G6) \hspace{1.2cm} $a\oplus b=\gyrab(b\oplus a)$

\noindent
for all $a,b\inn G$.
}
\end{definition}

First gyrogroup properties are studied in \cite[Chap.~1]{mybook05},
and more gyrogroup theorems are studied in \cite{mybook01,mybook02,mybook03}.
Thus, for instance, as in group theory, any gyrogroup possesses a unique
identity element which is both left and right, and any element of a gyrogroup
possesses a unique inverse.

In order to illustrate the power and elegance of the gyrogroup structure,
we solve below the two basic gyrogroup equations
\eqref{unk1} and \eqref{eqtghj03}.

Let us consider the gyrogroup equation
\begin{equation} \label{unk1}
\ab\op\xb=\bb
\end{equation}
in a gyrogroup $(G,\op)$ for the unknown $\xb$. If $\xb$ exists, then
by the right gyroassociative law \eqref{laws00}
and by \eqref{sprdhein}, we have
 \begin{equation} \label{solve1}
 \begin{split}
\xb &= \bz\op \xb \\
&=(\om \ab\op\ab)\op \xb \\
&= \om \ab\op (\ab\op \gyr[\ab,\om \ab]\xb) \\
&= \om \ab\op (\ab\op               \xb) \\
&= \om \ab\op\bb
 \end{split}
 \end{equation}
noting that $\gyr[\ab,\om\ab]$ is trivial by \eqref{sprdhein}.

Thus, if a solution to \eqref{unk1} exists, it must be given uniquely by 
\begin{equation} \label{sln1}
\xb=\om \ab\op\bb
\end{equation}

Conversely, if $\xb=\om \ab\op\bb$, then $\xb$ is indeed a solution
to \eqref{unk1} since
by the left gyroassociative law and \eqref{sprdhein} we have
 \begin{equation} \label{eq1ddf94}
 \begin{split}
\ab\op\xb&= \ab\op(\om \ab\op\bb)  \\
&= (\ab\op(\om \ab))\op{\gyr}[\ab,\om \ab] \bb  \\
&= \bz\op\bb \\
&= \bb
 \end{split}
 \end{equation}

Substituting the solution \eqref{sln1} in its equation \eqref{unk1}
and replacing $\ab$ by $\om \ab$
we recover the left cancellation law \eqref{eq01b} for Einstein addition
\begin{equation} \label{leftcanc}
\om \ab\op(\ab\op\bb) = \bb
\end{equation}

The gyrogroup operation (or, addition) of any gyrogroup
has an associated dual operation, called the {\it gyrogroup cooperation}
(or, {\it coaddition}), which is defined below.

\begin{definition}\label{defdual}
{\bf (The Gyrogroup Cooperation (Coaddition)).}
{\it
Let $(G,\op)$ be a gyrogroup with gyrogroup operation (or, addition) $\op$.
The gyrogroup cooperation (or, coaddition) $\sqp$ is a second
binary operation in $G$ given by the equation
\begin{equation} \label{eqdfhn01}
\ab\sqp \bb = \ab \op\gyr[\ab,\om \bb]\bb
\end{equation}
for all $\ab,\bb\in G$.
}
\end{definition}

Replacing $b$ by $\om b$ in \eqref{eqdfhn01}
we have the {\it cosubtraction} identity
\begin{equation} \label{eqdfhn01b}
\ab\sqm \bb := \ab\sqp(\om \bb) = \ab \om\gyr[\ab,\bb]\bb
\end{equation}
for all $\ab,\bb\in G$.

To motivate the introduction of the gyrogroup cooperation and
to illustrate the use of the left loop property $(G5)$, we solve
the equation
\begin{equation} \label{eqtghj03}
\xb\op \ab = \bb
\end{equation}
for the unknown $\xb$ in a gyrogroup $(G,\op)$.

Assuming that a solution
$\xb$ to \eqref{eqtghj03} exists, we have the following chain of equations
 \begin{equation} \label{eqtghj04}
 \begin{split}
 \xb &= \xb \op \bz \\
&= \xb \op( \ab \om  \ab )\\
&=( \xb \op  \ab )\op\gyr [ \xb , \ab ](\om  \ab )\\
&=( \xb \op  \ab )\om\gyr [ \xb , \ab ] \ab \\
&=( \xb \op  \ab )\om\gyr [ \xb \op  \ab , \ab ] \ab \\
&=\bb\om\gyr[\bb,\ab]\ab \\
&=\bb\sqm \ab
 \end{split}
 \end{equation}
where the gyrogroup cosubtraction, \eqref{eqdfhn01b},
which captures here an obvious analogy, comes into play.
Hence, if a solution $\xb$ to the gyrogroup equation \eqref{eqtghj03}
exists, it must be given uniquely by \eqref{eqtghj04}. One can show
that the latter is indeed a solution to \eqref{eqtghj03} \cite[Sec.~2.4]{mybook03}.

The gyrogroup cooperation is introduced into gyrogroups in order
to  capture useful analogies between gyrogroups and groups, and
to uncover duality symmetries
with the gyrogroup operation.
Thus, for instance, the gyrogroup cooperation uncovers the seemingly
missing right counterpart of the left cancellation law \eqref{eq01b},
giving rise to the right cancellation law,
\begin{equation} \label{eq01bb}
(\bb\sqm\ab) \op \ab = \bb
\end{equation}
for all $\ab,\ab$ in $G$, which is obtained by substituting the result of
\eqref{eqtghj04} into \eqref{eqtghj03}.

Remarkably, the right cancellation law \eqref{eq01bb} can be dualized,
giving rise to the dual right cancellation law
\begin{equation} \label{eq01bbb}
(\bb\om\ab) \sqp \ab = \bb
\end{equation}

As an example, and for later reference, we note that it follows from
the right cancellation law \eqref{eq01bb} that
\begin{equation} \label{hufns}
\db = (\bb\sqp \cb)\om\ab  \hspace{0.6cm} \Longleftrightarrow \hspace{0.6cm}
\bb\sqp \cb = \db\sqp\ab 
\end{equation}
for $\ab,\bb,\cb,\db$ in any gyrocommutative gyrogroup $(G,\op)$.

An elegant gyrocommutative gyrogroup identity that involves the
gyrogroup cooperation, verified in \cite[Theorem 3.12]{mybook03}, is
\begin{equation} \label{fromu1}
\ab\op(\bb\op\ab) = \ab\sqp(\ab\op\bb)
\end{equation}

A gyrogroup cooperation is commutative if and only if the gyrogroup is
gyrocommutative
\cite[Theorem 3.4]{mybook02}
\cite[Theorem 3.4]{mybook03}.
Hence, in particular, Einstein coaddition is commutative.
Indeed, Einstein coaddition, $\sqp$, in an Einstein gyrogroup $(\Rcn,\op)$,
defined in \eqref{eqdfhn01}, can be written as
\cite[Eq.~3.195]{mybook03}
\begin{equation} \label{eincoadd}
\begin{split}
\ub \sqp \vb &=
\frac{\gub+\gvb}{\gubs+\gvbs+\gub\gvb(1+\frac{\ub\ccdot\vb}{s^2})-1}
(\gub\ub + \gvb\vb)
\\[6pt] &=
\frac{\gub+\gvb}{(\gub+\gvb)^2-(\gamma_{\ub\om \vb}^{\phantom{O}} + 1)}
(\gub\ub + \gvb\vb)
\\[6pt] &=
2 \od \frac{\gub\ub+\gvb\vb}{\gub+\gvb}
\\[6pt] &=
2 \od \frac{\gub\ub+\gvb\vb}{2+(\gub-1)+(\gvb-1)}
\end{split}
\end{equation}
$\ub,\vb\in G$,
demonstrating that it is commutative, as expected.
The symbol $\od$ in \eqref{eincoadd} represents {\it scalar multiplication} so that,
for instance, $2\od\vb=\vb\op\vb$, for all $\vb$ in a gyrogroup $(G,\op)$,
as explained in Sec.~\ref{secc4} below.
It turns out that Einstein coaddition $\sqp$ is more than just a
commutative binary operation in the ball. Remarkably, it forms the
(hyperbolic) {\it gyroparallelogram addition law} in the ball,
illustrated in Fig.~\ref{fig190k4m}.

The extreme sides of \eqref{eincoadd} suggest that the application of
Einstein coaddition to three summands is given by the following
{\it gyroparallelepiped addition law}
\begin{equation} \label{eincoagf}
\ub \sqp_3 \vb \sqp_3 \wb
:=2 \od \frac{\gub\ub+\gvb\vb+\gwb\wb}{2+(\gub-1)+(\gvb-1)+(\gwb-1)}
\end{equation}
$\ub,\vb,\wb\in G$, the ternary operation $\sqp_3$ being
{\it Einstein coaddition of order three}.

Einstein coaddition \eqref{eincoagf} of three summands is
commutative and associative in the generalized sense that it is a
symmetric function of the summands.
The gyroparallelepiped that results from the
gyroparallelepiped law \eqref{eincoagf} is studied in detail
in \cite[Secs.~10.9--10.12]{mybook03}.

We may note that by \eqref{eincoadd}\,--\,\eqref{eincoagf} we have
$\ub \sqp_3 \vb \sqp_3 \zerb = \ub \sqp \vb$, as expected.
However, unexpectedly we have
$\ub \sqp_3 \vb \sqp_3 (\om\vb) \ne \ub$, in general.

The extension of \eqref{eincoagf}
to the Einstein coaddition of
$k$ summands, $k>3$ , is now straightforward, giving rise to the
higher dimensional {\it gyroparallelotope law} in $\Rcn$,
\begin{equation} \label{ejsnk06}
\vb_1 \sqp_k \vb_2 \sqp_k \, \ldots \, \sqp_k \vb_k
:= 2\od \frac{\sum_{i=1}^{k}\gvbi\vb_i}{2+\sum_{i=1}^{k}(\gvbi -1)}
\end{equation}
$\vb_k\in G$, $k\in\NN$, where $\sqp_k$ is a $k$-ary operation
called {\it Einstein coaddition of order} $k$.
An interesting study of parallelotopes in Euclidean geometry is found in
\cite{coxeterpolytopes}.

In the Euclidean limit $c\rightarrow\infty$,
(i) gamma factors tend to 1, and
(ii) the hyperbolic scalar multiplication, $\od$, of a
gyrovector (see Sec.~\ref{secs5}) by 2
tends to the common scalar multiplication of a vector by 2.
Hence, in the Euclidean limit, the right-hand side of \eqref{ejsnk06}
tends to the vector sum $\sum_{i=1}^{k}\vb_i$ in $\Rn$, as expected.

\section{Einstein Gyrovector Spaces}\label{secc4}

Let $k\od\vb$ be the Einstein addition of $k$ copies of $\vb\in \Rcn$,
that is
$k\od\vb=\vb\op\vb\ldots\op\vb$ ($k$ terms). Then,
it follows from Einstein addition \eqref{eq01} and straightforward algebra that
\cite{abslor92}
\begin{equation} \label{duhnd1}
k\od \vb = c \frac
 {\left(1+\displaystyle\frac{\| \vb \|}{c}\right)^k
- \left(1-\displaystyle\frac{\| \vb \|}{c}\right)^k} 
 {\left(1+\displaystyle\frac{\| \vb \|}{c}\right)^k
+ \left(1-\displaystyle\frac{\| \vb \|}{c}\right)^k}
\frac{\vb}{\| \vb\|}
\end{equation}

The definition of scalar multiplication in an Einstein gyrovector space
requires analytically continuing $k$ off the positive integers, thus
obtaining the following definition \cite{gaxioms}:

\begin{definition}\label{defgvspace}
An Einstein gyrovector space $(\Rcn,\op,\od)$
is an Einstein gyrogroup $(\Rcn,\op)$, $\Rcn\subset\Rcn$, with scalar multiplication
$\od$ given by the equation
 \begin{equation} \label{eqmlt03}
 \begin{split}
r\od\vb&= c \frac
 {\left(1+\displaystyle\frac{\|\vb\|}{c}\right)^r
- \left(1-\displaystyle\frac{\|\vb\|}{c}\right)^r}
 {\left(1+\displaystyle\frac{\|\vb\|}{c}\right)^r
+ \left(1-\displaystyle\frac{\|\vb\|}{c}\right)^r}
\frac{\vb}{\|\vb\|}\\[8pt]
&= c \tanh( r\,\tanh^{-1}\frac{\|\vb\|}{c})\frac{\vb}{\|\vb\|}
 \end{split}
 \end{equation}
where $r$ is any real number, $r\in\Rb$,
$\vb\in\Rcn$, $\vb\ne\bz$, and $r\od\bz=\bz$, and with
which we use the notation $\vb\od  r=r\od\vb$.
\end{definition}

Einstein gyrovector spaces are studied in \cite[Sec.~6.18]{mybook03}
and \cite{mybook04}.
Einstein scalar multiplication does not distribute over Einstein addition,
but it possesses other properties of vector spaces. For any
positive integer $n$, and for all real numbers $r,\ro,\rt\in\Rb$, and
$\vb\in\Rcn$, we have
\begin{alignat}{2}\label{scalarprp}
\notag
 n\od\vb&=\vb\op\dots\op\vb &&\qquad\text{$n$ terms}\\[3pt]
\notag
 (\ro+\rt)\od\vb&=\ro\od\vb\op\rt\od\vb
 &&\qquad\text{Scalar Distributive Law}\\[3pt]
\notag
(\ro\rt)\od\vb&=\ro\od(\rt\od\vb)
 &&\qquad\text{Scalar Associative Law} \\[3pt]
 \notag
r\od(\ro\od\vb\op\rt\od\vb)&=r\od(\ro\od\vb)\op r\od(\rt\od\vb)
 &&\qquad\text{Monodistributive Law}
 \notag
\end{alignat}
in any Einstein gyrovector space $(\Rcn,\op,\od)$.

Any Einstein gyrovector space $(\Rcn,\op,\od)$ inherits the common inner product
and the norm from its vector space $\Rn$. These turn out to be invariant under gyrations,
that is,
\begin{equation} \label{eqwiuj01}
\begin{split}
\gyr[\ab,\bb]\ub \ccdot \gyr[\ab,\bb]\vb &= \ub \ccdot \vb \\[3pt]
\| \gyr[\ab,\bb]\vb \| &= \| \vb \|
\end{split}
\end{equation}
for all $\ab,\bb,\ub,\vb\inn\Rcn$.

Unlike vector spaces, Einstein gyrovector spaces $(\Rcn,\op,\od)$ do not
possess the distributive law since, in general,
\begin{equation} \label{eqwdkj01}
r\od(\ub\op\vb) \ne r\od\ub\op r\od\vb
\end{equation}
for $r\inn\Rb$ and $\ub,\vb\inn\Rcn$.
One might suppose that there is a price to pay in mathematical regularity
when replacing ordinary vector addition with Einstein addition,
but  this is not the case as demonstrated in \cite{mybook01,mybook02,mybook03},
and as noted by S.~Walter in \cite{walterrev2002}.

Owing to the break down of the distributive law in gyrovector spaces, the
following gyrovector space identity, called the {\it Two-Sum Identity}
\cite[Theorem 6.7]{mybook03}, proves useful:
\begin{equation} \label{hyets1}
2\od(\ub\op\vb)=\ub\op(2\od\vb\op\ub)
\end{equation}

In full analogy with the common Euclidean distance function,
Einstein addition admits the {\it gyrodistance} function
\begin{equation} \label{tkcflsen}
d_\op(A,B) = \|\om A \op B\|
\end{equation}
that obeys the gyrotriangle inequality
\cite[Theorem 6.9]{mybook03}
\begin{equation} \label{rif01}
d_\op( A , B ) \le d_\op( A , P ) \op d_\op( P , B )
\end{equation}
for any points $A,B,P\in\Rcn$ in an Einstein gyrovector space $(\Rcn,\op,\od)$.
The gyrodistance function is invariant under the group of motions of its
Einstein gyrovector space, that is, under {\it left gyrotranslations} and
rotations of the space \cite[Sec.~4]{mybook03}.
The gyrotriangle inequality \eqref{rif01} reduces to a corresponding
gyrotriangle equality,
\begin{equation} \label{rif01s}
d_\op( A , B ) = d_\op( A , P ) \op d_\op( P , B )
\end{equation}
if and only if point $ P $ lies between points $ A $ and $ B $, that is,
point $ P $ lies on the gyrosegment $ A  B $, as shown in
Fig.~\ref{fig163b3enm}.
Accordingly, the gyrodistance function is gyroadditive on gyrolines, as
demonstrated in \eqref{rif01s} and
illustrated graphically in Fig.~\ref{fig163b3enm}.

Furthermore, the Einstein gyrodistance function \eqref{tkcflsen} in any
$n$-dimensional Einstein gyrovector space $(\Rcn,\op,\od)$
possesses a familiar Riemannian line element. It
gives rise to the Riemannian line element $ds_e^2$ of the Einstein gyrovector space
with its {\it gyrometric} \eqref{tkcflsen},
\begin{equation} \label{eqdelta01fen}
\begin{split}
ds_e^2 &=  \|(\xb+d\xb) \om \xb\|^2 \\[4pt]
&= \frac{c^2} {c^2-\xb^2}d\xb^2 +
  \frac{c^2} {(c^2-\xb^2)^2}(\xb\ccdot d\xb)^2
\end{split}
\end{equation}
$\xb\in\Rcn$, where $d\xb^2 = d\xb\ccdot d\xb$,
as shown in \cite[Theorem 7.6]{mybook03}.

Remarkably, the Riemannian line element $ds_e^2$ in \eqref{eqdelta01fen} turns out
to be the well-known line element that the
Italian mathematician Eugenio Beltrami introduced in 1868
in order to study hyperbolic geometry by a Euclidean disc model,
now known as the Beltrami-Klein disc
\cite[p.~220]{mccleary94}\cite{barrett00}.
An English translation of his historically significant 1868 essay
on the interpretation of non-Euclidean geometry is found in
\cite{stillwell96}. The significance of Beltrami's 1868 essay
rests on the generally known fact that
it was the first to offer a concrete interpretation of hyperbolic geometry
by interpreting ``straight lines'' as geodesics on a surface of a constant
negative curvature.
Beltrami, thus, constructed a Euclidean disc model of the
hyperbolic plane \cite{mccleary94} \cite{stillwell96}, which now bears his name
along with the name of Klein.

We have thus found that the Beltrami-Klein
ball model of hyperbolic geometry is regulated algebraically by Einstein gyrovector spaces
with their gyrodistance function \eqref{tkcflsen}
and Riemannian line element \eqref{eqdelta01fen},
just as the standard model of Euclidean geometry is regulated algebraically by vector spaces
with their Euclidean distance function
and the Riemannian line element $ds^2 = d\xb^2$.

\begin{figure}[t]  
 \sidebyside {       
 \psfrag{pa}  {$A,\hspace{0.1cm}t=0$}
 \psfrag{pb}  {$B,\hspace{0.1cm}t=1$}
\psfrag{formula01}[]{${\rm The~Einstein~Gyroline} L_{^{AB}}$}
\psfrag{formula02}[]{${\rm through~the~points}~ A~{\rm and}~B$}
\psfrag{formula03}[]{$\boxed{A\op (\om  A\op  B)\od  t}$}
\psfrag{formula04}[]{$ -\infty < t < \infty $}
\psfrag{formula05}[]{$\op=\ope$}
 \includegraphics[width=0.45\textwidth]{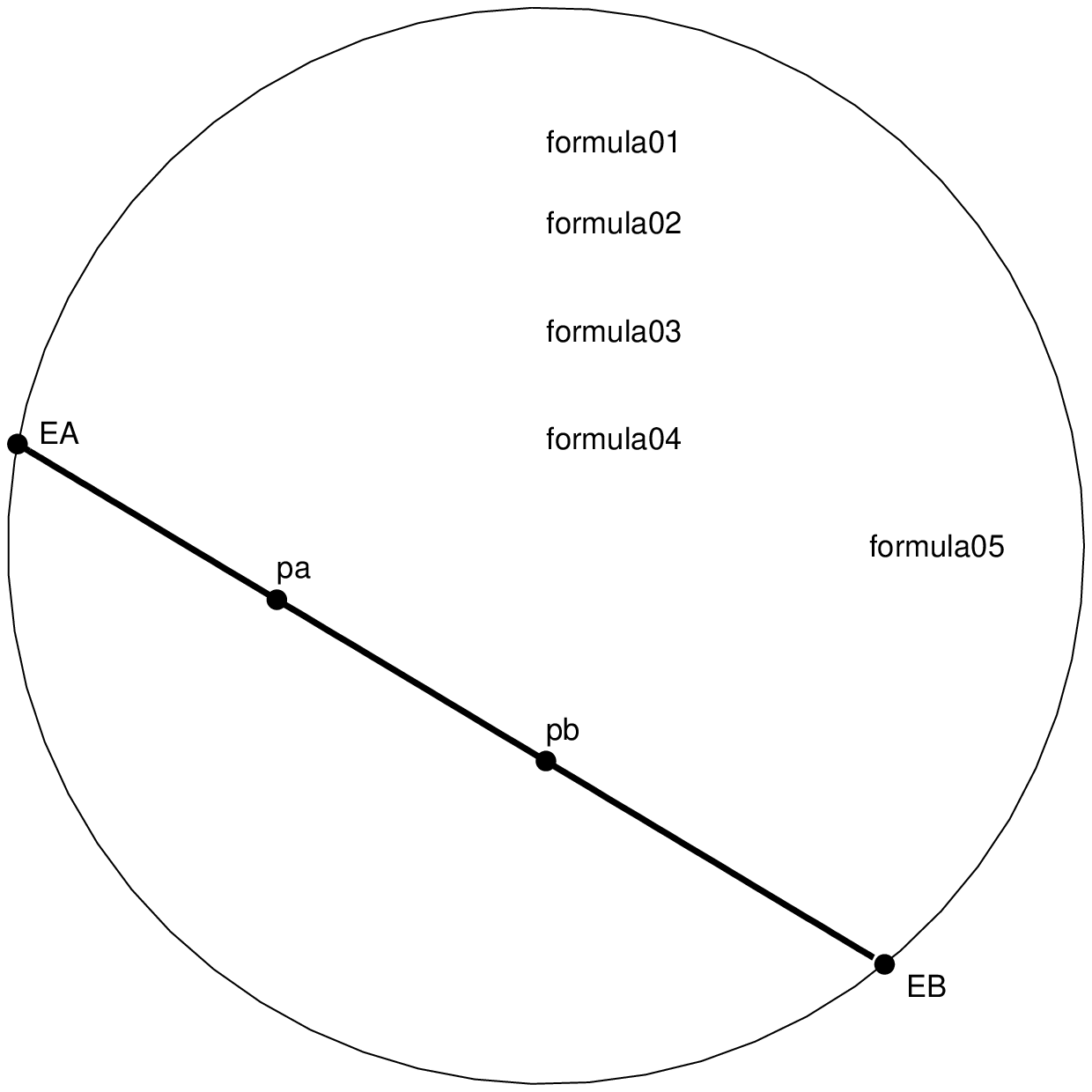}
\caption[The Einstein Gyroline]{
The unique gyroline $L_{^{AB}}$ in an Einstain gyrovector space $(\Rcn,\op ,\od )$
through two given points $A$ and $B$.
The case of the Einstain gyrovector plane, when $\Rcn=\Rctwou$ is the
real open unit disc, is shown.
\label{fig174b3enm}}}
  {
\psfrag{a}[]{$A$}
\psfrag{b}[]{$\!\! B$}
\psfrag{mab}{$M_{^{AB}}^{\phantom{O}} \!=\! \half\od (A\sqp \! B)$}
\psfrag{pab}{$P$}
\psfrag{formula00}[]{$d_{\op }(A,P) \op  d_{\op }(P,B)=d_{\op }(A,B)$}
\psfrag{formula01}[]{$\boxed{A\op (\om  A\op  B)\od  t}$}
\psfrag{formula02}[]{$0\le t \le1$}
\psfrag{formula05}[]{$\op=\ope$}

 \includegraphics[width=0.45\textwidth]{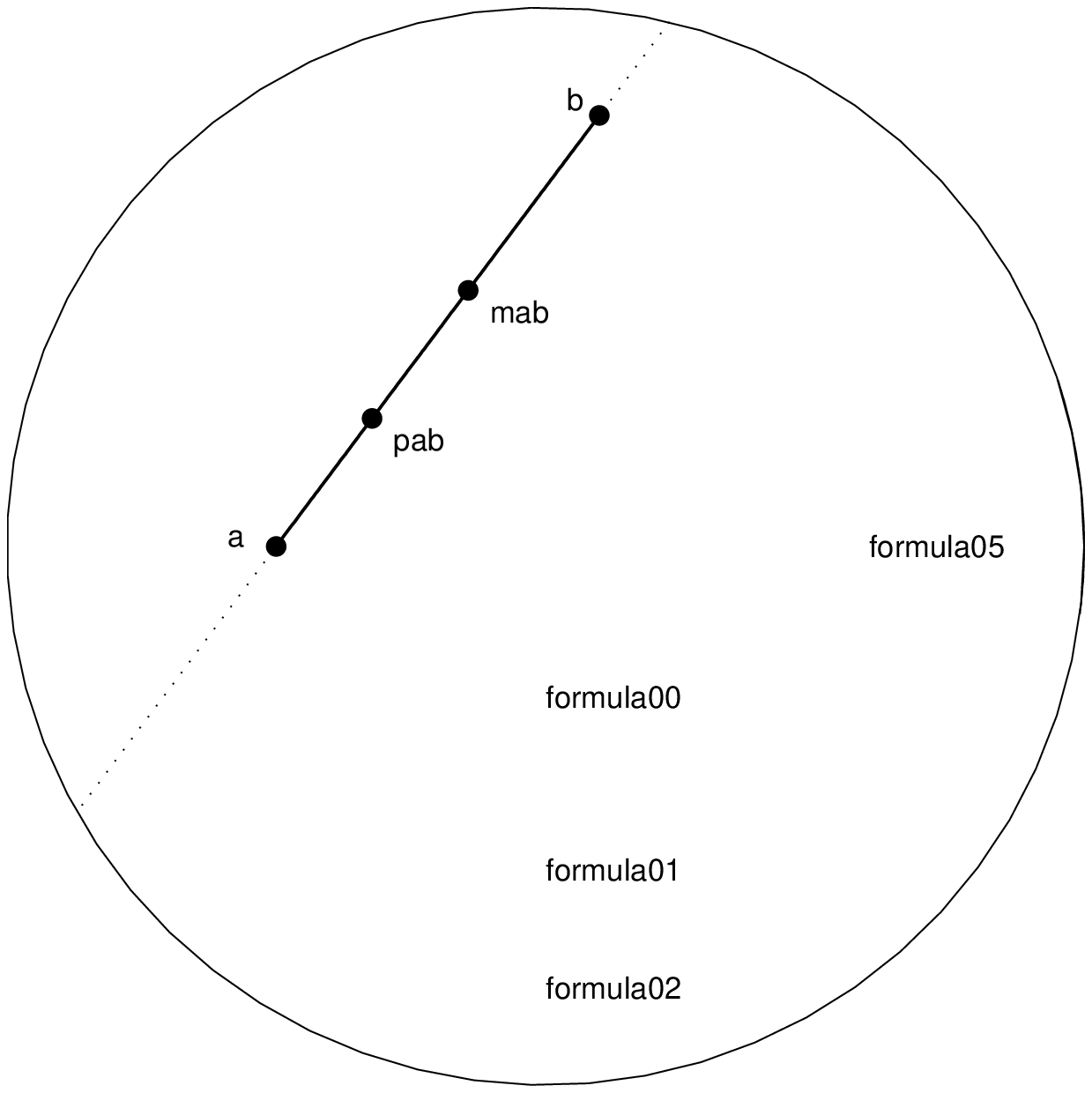}
\caption{
The gyrosegment $AB$ that links
the points $A$ and $B$ in $(\Rcn,\op ,\od )$, with one of its generic points
$P$ and its gyromidpoint $M_{^{AB}}$. The point $P$ lies between $A$ and $B$
and, hence, obeys the gyrotriangle equality, \eqref{rif01s}.
\label{fig163b3enm}} }
\end{figure}


In full analogy with Euclidean geometry,
the unique Einstein gyroline $\LAB$, Fig.~\ref{fig174b3enm},
that passes through two given
points $A$ and $B$ in an Einstein gyrovector space $\Rcn=(\Rcn,\op ,\od)$
is given by the parametric equation
 \begin{equation} \label{eqrfgh01en}
\LAB(t) = A\op (\om A \op  B)\od  t
 \end{equation}
with the parameter $t\in\Rb$.
The gyroline $\LAB$ passes through the point $A$ when $t=0$ and,
owing to the left cancellation law \eqref{eq01b}, it
passes through the point $B$ when $t=1$.

Einstein gyrolines in the ball $\Rcn$ are chords of the ball,
as shown in Fig.~\ref{fig174b3enm}. These chords of the ball turn out to be the
familiar geodesics of the Beltrami-Klein ball model of hyperbolic geometry
\cite{mccleary94}.
Accordingly, Einstein gyrosegments are Euclidean segments, as shown in
Fig.~\ref{fig163b3enm}.
The result that Einstein gyrosegments are Euclidean segments is well exploited
in \cite{mybook06,mybook05} in the use of hyperbolic barycentric coordinates
for the determination of various hyperbolic triangle centers.
It enables one to determine points of intersection of gyrolines
by common methods of linear algebra.

The gyromidpoint $\MAB$ of gyrosegment $AB$, shown in Fig.~\ref{fig163b3enm}, is the
unique point of the gyrosegment that satisfies the equation
$d_\op(\MAB,A) = d_\op(\MAB,B)$. It is given by each of the following equations
\cite[Theorem 3.33]{mybook04},
\begin{equation} \label{eqgyromid}
\MAB = A\op (\om A \op  B)\od \half
     = \frac{\gamma_{_A}A+\gamma_{_B}B}{\gamma_{_A}+\gamma_{_B}} =
\half \od (A\sqp B)
\end{equation}
in full analogy with Euclidean midpoints, shown in Fig.~\ref{fig190k4eucm}.
One may note that the extreme right equation in \eqref{eqgyromid}
appears in \eqref{eincoadd} in an equivalent form.

The endpoints of a gyroline in an Einstein gyrovector space $(\Rcn,\op,\od)$
are the points where the gyroline approaches the boundary of
the ball $\Rcn$.
Following \eqref{eqrfgh01en}, the endpoints $E_{^A}$ and $E_{^B}$
of the gyroline $\LAB(t)$ in Fig.~\ref{fig174b3enm} are
\begin{equation} \label{kulib}
\begin{split}
E_{^A} &= \lim_{t\rightarrow -\infty}
\{A\op (\om A \op  B)\od  t\}
\\
E_{^B} &= \lim_{t\rightarrow \phantom{-}\infty}
\{A\op (\om A \op  B)\od  t\}
\end{split}
\end{equation}
Explicit expressions for the gyroline endpoints in Einstein gyrovector spaces
are presented in \eqref{grut06}, p.~\pageref{grut06}.

\section{Vectors and Gyrovectors}\label{secs5}


\begin{figure}[t]  
 \sidebyside {       
\psfrag{A}[]{$P$}
\psfrag{B}[]{$Q$}
\psfrag{Ap}[]{$R$}
\psfrag{Bp}[]{$S$}
\psfrag{formula00}[]{$-P+Q=-R+S$}
\psfrag{formula01}[]{$d(P,Q)=\|-P+Q\|$}
\psfrag{---chets1}[]{\lower-1.2ex \hbox {\footnotesize{$\blacktriangleright$}}}
\psfrag{---chets2}[]{\lower-1.2ex \hbox {\footnotesize{$\blacktriangleright$}}}
 \includegraphics[width=0.5\textwidth]{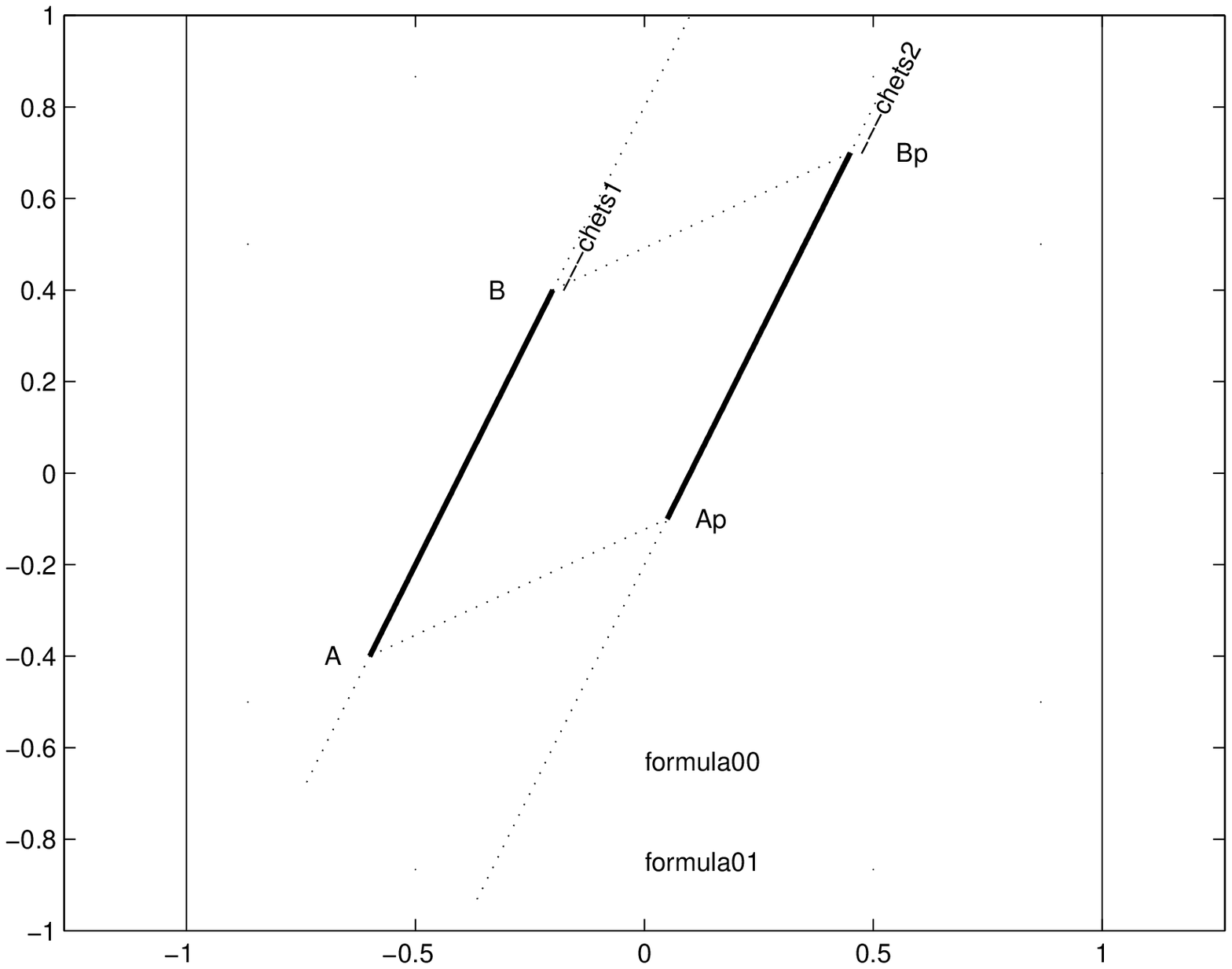}
\caption{
Two equivalent vectors in a Euclidean vector plane $(\Rtwo,+,\ccdot)$.
The two vectors are parallel and have equal values and, hence, equal lengths.
\label{fig221a6m}}}
  {
\psfrag{A}[]{$P$}
\psfrag{B}[]{$Q$}
\psfrag{Ap}[]{$R$}
\psfrag{Bp}[]{$S$}
\psfrag{formula00}[]{$\om P\op Q=\om R\op S$}
\psfrag{formula01}[]{$d(P,Q)=\|\om P\op Q\|$}
\psfrag{---chets1}[]{\lower-1.2ex \hbox {\footnotesize{$\blacktriangleright$}}}
 \psfrag{---chets2}[]{\lower-1.0ex \hbox {\footnotesize{$\blacktriangleright$}}}
 \includegraphics[width=0.5\textwidth]{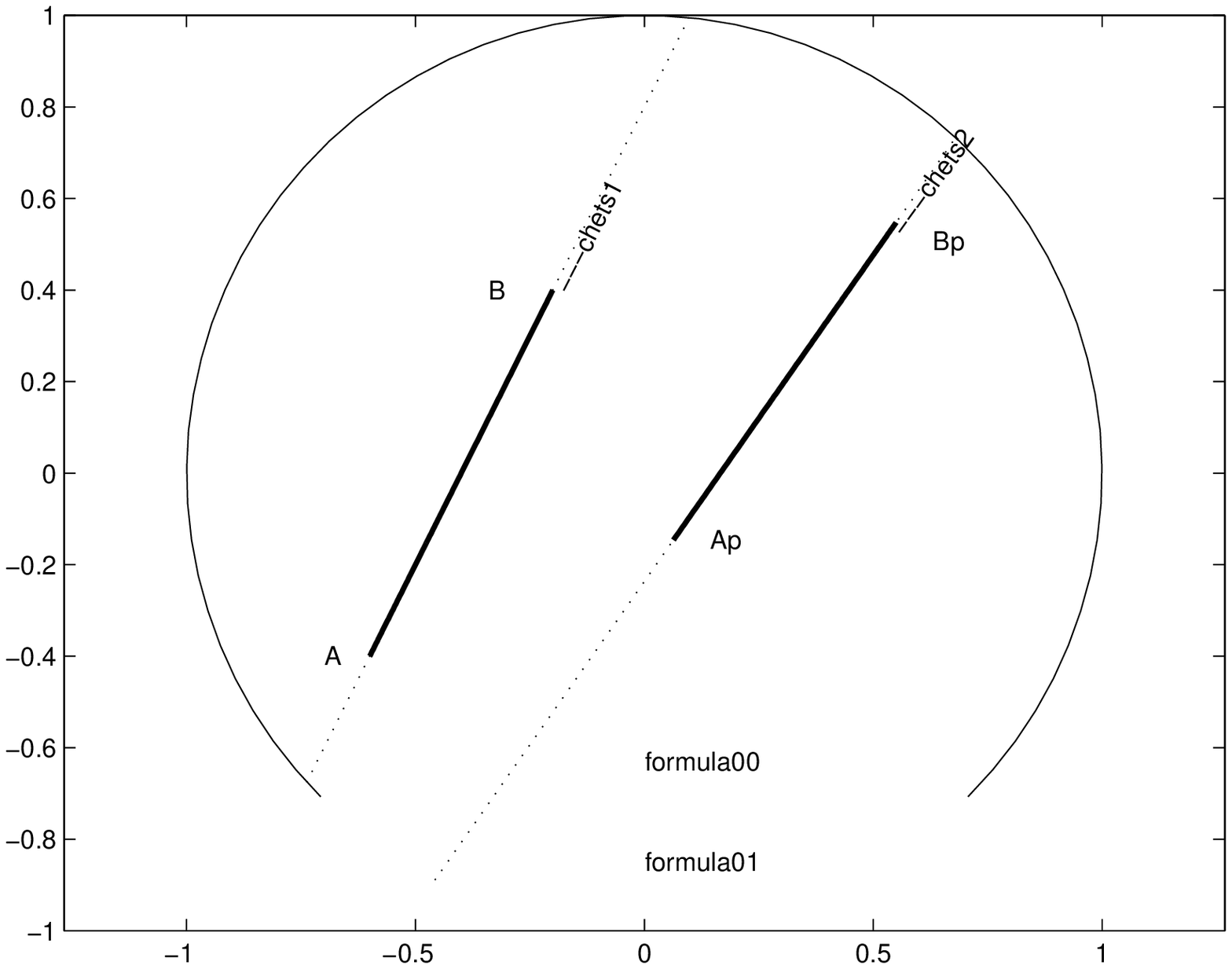}
\caption[]{
Two equivalent gyrovectors in an Einstein gyrovector plane $(\Rctwo,\op,\od)$.
The two gyrovectors have equal values and, hence, equal gyrolengths.
\label{fig221b6m}}}
\end{figure}

Elements of a real inner product space $\vi=(\vi,+,\ccdot)$, called
points and denoted by capital italic letters, $A,B,P,Q,$ etc,
give rise to vectors in $\vi$, denoted by bold roman lowercase letters $\ub,\vb,$ etc.
Any two ordered points $P,Q\inn\vi$ give rise to a unique rooted vector
$\vb\inn\vi$, rooted at the point $P$. It has a tail at the point $P$
and a head at the point $Q$, and it has the value $-P+Q$,
\begin{equation} \label{eq2rhkd01}
\vb =-P+Q
\end{equation}
The length of the rooted vector $\vb = -P+Q$ is the distance between the points
$P$ and $Q$, given by the equation
\begin{equation} \label{eq2rhkd02}
\|\vb\| = \|-P+Q\|
\end{equation}

Two rooted vectors $-P+Q$ and $-R+S$ are equivalent if they have the same
value, that is,
\begin{equation} \label{fksnc}
-P+Q ~~\thicksim~~-R+S \hspace{1.0cm} {\rm if~and~only~if} \hspace{1.0cm} -P+Q=-R+S
\end{equation}
The relation $\thicksim$ in \eqref{fksnc} between rooted vectors is
reflexive, symmetric and transitive, so that it is an equivalence relations that gives
rise to equivalence classes of rooted vectors.

Two equivalent rooted vectors in a Euclidean vector plane are shown in Fig.~\ref{fig221a6m}.
Being equivalent in Euclidean geometry, the two vectors in Fig.~\ref{fig221a6m}
are parallel and they possess equal lengths.

To liberate rooted vectors from
their roots we define
a {\it vector} to be an equivalence class of rooted vectors.
The vector $-P+Q$ is thus a representative of all rooted vectors with value $-P+Q$.
Accordingly, the two vectors in Fig.~\ref{fig221a6m} are equal.

A point $P\inn\vi$ is identified with the vector $-O+P$,
$O$ being the arbitrarily selected origin of the space $\vi$.
Hence, the algebra of vectors can be applied to points as well.
Naturally, geometric and physical properties regulated by a vector space are
independent of the choice of the origin.

Let $A,B,C\in\vi$ be three non-collinear points, and let
\begin{equation} \label{ksfr1}
\begin{split}
\ub &= -A+B \\
\vb &= -A+C
\end{split}
\end{equation}
be two vectors in $\vi$ that possess the same tail, $A$.
Furthermore, let $D$ be a point of $\vi$ given by the
parallelogram condition
\begin{equation} \label{ksfr2}
D = B+C-A
\end{equation}

  
\begin{figure}[t]  
 \centering         
\psfrag{----chets}[]{\lower-1.0ex \hbox {\footnotesize{$\blacktriangleright$}}}
\psfrag{mab}[]{$\hspace{0.66cm}\MABDC$}
\psfrag{O}[]{$\phantom{O}$}
\psfrag{a}[]{$A$}
\psfrag{b}[]{$B$}
\psfrag{c}[]{$C$}
\psfrag{d}[]{$D$}
\psfrag{formula00}{${\rm The~Parallelogram}$} 
\psfrag{formula00a}{$\hspace{-.1cm}{\rm Condition:}~D = B+C-A$} 
\psfrag{formula01}{$\MAD = \half(A+D)$}
\psfrag{formula02}{$\MBC = \half(A+C)$}
\psfrag{formula03}{$\MABDC = \frac{A + B + C + D}{4}$}
 \psfrag{formula03a}{$\hspace{-.66cm}\MABDC=\MAD=\MBC$}
 \psfrag{formula04}{$-C+D = -A+B$}
\psfrag{formula05}{$-B+D =-A+C$}
\psfrag{fig190j}[]{$\boxed{(-A+B)+(-A+C)=-A+D}$}
\psfrag{formula06}[]{$\boxed{\ub+\vb=\wb}$}
\psfrag{text1}[]{$\hspace{0.20cm}\ub=-A+B$}
\psfrag{text2}[]{$\vb=-A+C$}
\psfrag{text3}[]{$\wb=-A+D$}
 \includegraphics[width=9cm]{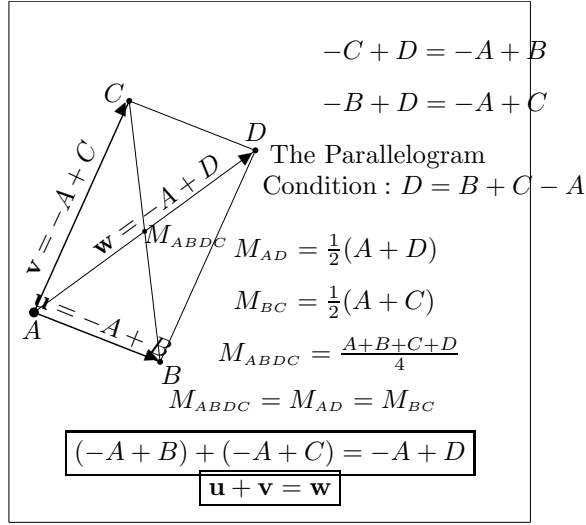}
\caption{
The Euclidean parallelogram and its addition law
in a Euclidean vector plane $(\Rtwo,+,\ccdot)$.
The diagonals $AD$ and $BC$ of parallelogram $ABDC$
intersect each other at their midpoints.
The midpoints of the diagonals $AD$ and $BC$ are, respectively,
$M_{^{AD}}$ and $M_{^{BC}}$, each of which coincides with the
parallelogram center $M_{^{ABDC}}$.
\label{fig190k4eucm}}
\end{figure}


The quadrangle (also known as a quadrilateral; see \cite[p.~52]{coxetergreitzer67})
$ABDC$ turns out to be a parallelogram in Euclidean geometry,
Fig.~\ref{fig190k4eucm}, since
its two diagonals, $AD$ and $BC$, intersect at their midpoints, that is,
\begin{equation} \label{ksfr3}
\half(A+D) = \half(B+C)
\end{equation}
Clearly, the midpoint equality \eqref{ksfr3} is equivalent to the
parallelogram condition \eqref{ksfr2}.

The vector addition of the vectors $\ub$ and $\vb$
that generate the parallelogram $ABDC$, according to \eqref{ksfr1},
gives the vector $\wb$ by the parallelogram addition law, Fig.~\ref{fig190k4eucm},
\begin{equation} \label{ksfr4}
\wb := -A+D = (-A+B) + (-A+C) = \ub+\vb
\end{equation}
Here, by definition,
$\wb$ is the vector formed by the diagonal $AD$ of the parallelogram $ABDC$,
as shown in Fig.~\ref{fig190k4eucm}.

Vectors in the space $\vi$
are, thus, equivalence classes of ordered pairs of points, Fig.~\ref{fig221a6m},
which add according to the parallelogram law, Fig.~\ref{fig190k4eucm}.

Gyrovectors emerge in an Einstein gyrovector space $(\vc,\op,\od)$ in a way fully analogous
to the way vectors emerge in the space $\vi$,
where $\vc$ is the $c$-ball of the space $\vi$, \eqref{eqsball}.

Elements of $\vc$, called
points and denoted by capital italic letters, $A,B,P,Q,$ etc,
give rise to gyrovectors in $\vc$, denoted by bold roman lowercase letters $\ub,\vb,$ etc.
Any two ordered points $P,Q\inn\vc$ give rise to a unique rooted gyrovector
$\vb\inn\vc$, rooted at the point $P$. It has a tail at the point $P$
and a head at the point $Q$, and it has the value $\om P\op Q$,
\begin{equation} \label{eq2rhkd01g}
\vb =\om P\op Q
\end{equation}
The gyrolength of the rooted gyrovector $\vb = \om P\op Q$ is the gyrodistance between the points
$P$ and $Q$, given by the equation
\begin{equation} \label{eq2rhkd02g}
\|\vb\| = \|\om P\op Q\|
\end{equation}

Two rooted gyrovectors $\om P\op Q$ and $\om R\op S$ are equivalent
if they have the same value, that is,
\begin{equation} \label{fksnd}
\om P\op Q ~~\thicksim~~ \om R\op S \hspace{1.0cm} {\rm if~and~only~if} \hspace{1.0cm}
\om P\op Q= \om R \op S
\end{equation}
The relation $\thicksim$ in \eqref{fksnd} between rooted gyrovectors is
reflexive, symmetric and transitive, so that it is an equivalence relation that gives
rise to equivalence classes of rooted gyrovectors.

Two equivalent rooted gyrovectors in an Einstein gyrovector plane are shown in
Fig.~\ref{fig221b6m}.
Being equivalent in hyperbolic geometry, the two gyrovectors in Fig.~\ref{fig221b6m}
possess equal gyrolengths.

To liberate rooted gyrovectors from
their roots we define
a {\it gyrovector} to be an equivalence class of rooted gyrovectors.
The gyrovector $\om P\op Q$ is thus a representative of all rooted gyrovectors
with value $\om P\op Q$.
Accordingly, the two gyrovectors in Fig.~\ref{fig221b6m} are equal.

A point $P$ of a gyrovector space $(\vc,\op,\od)$ is identified with the gyrovector $\om O\op P$,
$O$ being the arbitrarily selected origin of the space $\vc$.
Hence, the algebra of gyrovectors can be applied to points as well.
Naturally, geometric and physical properties regulated by a gyrovector space are
independent of the choice of the origin.

Let $A,B,C\in\vc$ be three non-gyrocollinear points of an Einstein gyrovector space
$(\vc,\op,\od)$, and let
\begin{equation} \label{ktfr1}
\begin{split}
\ub &= \om A\op B \\
\vb &= \om A\op C
\end{split}
\end{equation}
be two gyrovectors in $\vi$ that possess the same tail, $A$.
Furthermore, let $D$ be a point of $\vc$ given by the
gyroparallelogram condition
\begin{equation} \label{ktfr2}
D = (B\sqp C)\om A
\end{equation}
Then, the gyroquadrangle $ABDC$ is a gyroparallelogram in the
Beltrami-Klein ball model of hyperbolic geometry in the sense
that its two gyrodiagonals, $AD$ and $BC$, intersect at their gyromidpoints, that is,
\begin{equation} \label{ktfr3}
\half\od(A\sqp D) = \half\od(B\sqp C)
\end{equation}
as illustrated in Fig.~\ref{fig190k4m}.
Clearly by \eqref{hufns},
the gyromidpoint equality \eqref{ktfr3} is equivalent to the
gyroparallelogram condition \eqref{ktfr2}.

The gyrovector addition of the gyrovectors $\ub$ and $\vb$
that generate the gyroparallelogram $ABDC$ gives the gyrovector $\wb$
by the gyroparallelogram addition law, Fig.~\ref{fig190k4m},
\begin{equation} \label{ktfr4}
\wb := \om A \op D = (\om A \op B) \sqp (\om A \op C) =: \ub \sqp \vb
\end{equation}
Here, by definition, $\wb$ is the gyrovector formed by the gyrodiagonal $AD$
of the gyroparallelogram $ABDC$.
The gyrovector identity in \eqref{ktfr4} is explained in
\eqref{ryuf1} below.

Gyrovectors in the ball $\vc$
are, thus, equivalence classes of ordered pairs of points, Fig.~\ref{fig221b6m},
which add according to the gyroparallelogram law, Fig.~\ref{fig190k4m}.

\section{Gyroparallelogram -- The Hyperbolic Parallelogram} \label{secgpar}

  
\begin{figure}[t]  
 \centering         
\psfrag{----chets}[]{\lower-1.0ex \hbox {\footnotesize{$\blacktriangleright$}}}
\psfrag{mab}[]{$\hspace{0.66cm}\MABDC$}
\psfrag{O}[]{$\phantom{O}$}
\psfrag{a}[]{$A$}
\psfrag{b}[]{$B$}
\psfrag{c}[]{$C$}
\psfrag{d}[]{$D$}
\psfrag{formula00}{${\rm The~Gyroparallelogram}$} 
\psfrag{formula00a}{${\rm Condition:}~~~~D = (B\sqp C)\om A$} 
\psfrag{formula01}{$\MAD = \frac{\gamma_{_A}A+\gamma_{_D}D}{\gamma_{_A}+\gamma_{_D}}
 = \half\od(A\sqp D)$}
\psfrag{formula02}{$\MBC = \frac{\gamma_{_B}B+\gamma_{_C}C}{\gamma_{_B}+\gamma_{_C}}
 = \half\od(A\sqp C)$}
\psfrag{formula03}{$\MABDC = \frac{\gmA A +\gmB B +\gmC C +\gmD D}{\gmA+\gmB+\gmC+\gmD}$}
 \psfrag{formula03a}{$\hspace{-.66cm}\MABDC=\MAD=\MBC$}
 \psfrag{formula04}[][][0.8]{$\hspace{5.4cm}\om C \op D =\gyr[C,\om B]\gyr[B,\om A](\om A\op B)$}
 \psfrag{formula05}[][][0.8]{$\hspace{5.4cm}\om B \op D =\gyr[B,\om C]\gyr[C,\om A](\om A\op C)$}
\psfrag{fig190j}[]{$\boxed{(\om A \op B )\sqp(\om A \op C )=\om A \op D }$}
\psfrag{formula06}[]{$\boxed{\ub\sqp\vb=\wb}$}
\psfrag{text1}[]{$\hspace{0.20cm}\ub=\om A\op B$}
\psfrag{text2}[]{$\vb=\om A\op C$}
\psfrag{text3}[]{$\wb=\om A\op D$}
 \includegraphics[width=9cm]{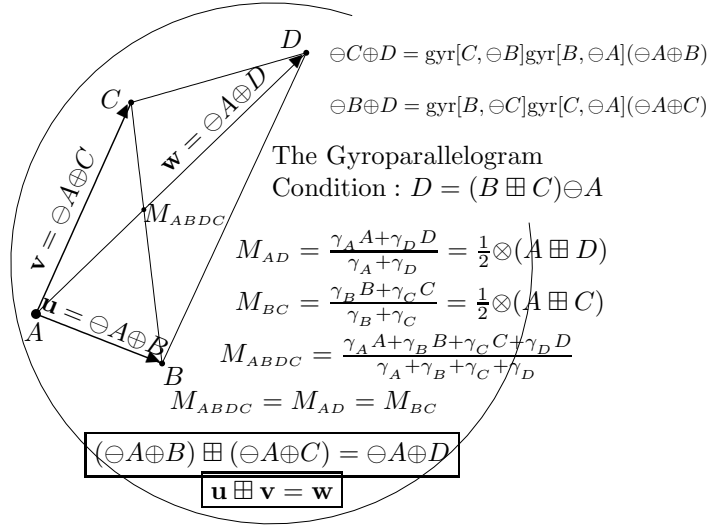}
\caption{
The Einstein gyroparallelogram and its addition law
in an Einstein gyrovector plane $(\Rctwo,\op,\od)$.
The gyrodiagonals $AD$ and $BC$ of gyroparallelogram $ABDC$
intersect each other at their gyromidpoints.
Detailed studies of the gyroparallelogram and its extension to higher dimensional
gyroparallelepipeds are presented in \cite{mybook02,mybook03}.
The gyroparallelogram addition law plays an important role in the
gyrovector space approach to hyperbolic geometry, studied in \cite{mybook03,mybook04}.
The gyromidpoints of the gyrodiagonals $AD$ and $BC$ are, respectively,
$M_{^{AD}}$ and $M_{^{BC}}$, each of which coincides with the
gyroparallelogram gyrocenter $M_{^{ABDC}}$.
The analogies that this figure shares with Fig.~\ref{fig190k4eucm} are obvious.
Along these analogies there is a remarkable disanalogy.
(i) Newton velocity addition, +, and the parallelogram addition, +,
in Fig.~\ref{fig190k4eucm} are identically the same binary operations in $\Rn$.
In contrast
(ii) Einstein velocity addition, $\op$, and its resulting gyroparallelogram addition, $\sqp$,
in this figure are two different binary operations in the ball $\Rcn$.
This disanalogy raises the question as to whether uniform relativistic velocities
in the Universe
are added according to the noncommutative Einstein velocity addition, \eqref{eq01},
or according to the commutative Einstein gyroparallelogram addition,
$\sqp$ in \eqref{eqdfhn01}.
\label{fig190k4m}}
\end{figure}


In Euclidean geometry a parallelogram is a
quadrangle the two diagonals of which intersect at their midpoints.
In full analogy, in hyperbolic geometry a gyroparallelogram is a gyroquadrangle
the two gyrodiagonals of which intersect at their gyromidpoints,
as shown in Fig.~\ref{fig190k4m}.
Accordingly, if $A$, $B$ and $C$ are any three non-gyrocollinear points
(that is, they do not lie on a gyroline) in an Einstein gyrovector space,
and if a fourth point $D$ is given by the {\it gyroparallelogram condition}
\begin{equation} \label{hkmdnf1}
D = (B\sqp C)\om A
\end{equation}
then the gyroquadrangle $ABDC$ is a gyroparallelogram, as shown in
Fig.~\ref{fig190k4m}.

Indeed, the two gyrodiagonals of gyroquadrangle $ABDC$ are the gyrosegments
$AD$ and $BC$, shown in Fig.~\ref{fig190k4m},
the gyromidpoints of which coincide, that is,
\begin{equation} \label{hkmdnf2}
\half\od(A\sqp D) =  \half\od(B\sqp C)
\end{equation}
where, by \eqref{hufns},
the result in \eqref{hkmdnf2} is equivalent to
the gyroparallelogram condition \eqref{hkmdnf1}.

Let $ABC$ be a gyrotriangle in an Einstein gyrovector space $(\Rcn,\op,\od)$
and let $D$ be the point that augments gyrotriangle $ABC$ into the
gyroparallelogram $ABDC$, as shown in Fig.~\ref{fig190k4m}.
Then, $D$ is determined uniquely by
the gyroparallelogram condition \eqref{hkmdnf1}, obeying
the gyroparallelogram addition law \cite[Theorem 5.5]{mybook05}
\begin{equation} \label{ryuf1}
(\om A\op B) \sqp (\om A\op C) = (\om A\op D)
\end{equation}
shown in Fig.~\ref{fig190k4m}.
In full analogy with the parallelogram addition law of vectors
in Euclidean geometry, \eqref{ksfr4},
the gyroparallelogram addition law \eqref{ryuf1} of gyrovectors in
hyperbolic geometry can be written as
\begin{equation} \label{ryuf2}
\ub \sqp \vb = \wb
\end{equation}
where $\ub,\vb$ and $\wb$ are the {\it gyrovectors}
\begin{equation} \label{ryuf3}
\begin{split}
\ub &= \om A \op B \\
\vb &= \om A \op C \\
\wb &= \om A \op D
\end{split}
\end{equation}
which emanate from the point $A$ \cite[Chap.~5]{mybook03}.

In his 1905 paper that founded the special theory of relativity \cite{einstein05},
Einstein noted that his velocity addition does not satisfy the
Euclidean parallelogram law:
\begin{quotation}
``Das Gesetz vom Parallelogramm der Geschwindigkeiten gilt also nach
unserer Theorie nur in erster Ann\"aherung.''
\begin{flushright}
A.~Einstein \cite{einstein05}
\end{flushright}
\end{quotation}
[English translation: Thus the law of velocity parallelogram is valid according to our
theory only to a first approximation.]

Indeed, Einstein velocity addition, $\op$, is noncommutative and does not give rise
to an exact ``velocity parallelogram'' in Euclidean geometry.
However, as we see in Fig.~\ref{fig190k4m}, Einstein velocity {\it coaddition}, $\sqp$,
which is commutative, does give rise to an exact
``velocity gyroparallelogram'' in hyperbolic geometry.

The breakdown of commutativity in Einstein velocity addition law seemed
undesirable to the famous mathematician \'Emile Borel.
Borel's resulting attempt to ``repair'' the seemingly ``defective'' Einstein velocity
addition in the years following 1912 is described by Walter in
\cite[p.~117]{walter99b}.
Here, however, we see that there is no need to repair Einstein velocity
addition law for being noncommutative since, despite of being noncommutative,
it gives rise to the gyroparallelogram law of gyrovector addition, which
turns out to be commutative.
The compatibility of the gyroparallelogram addition law of Einsteinian velocities
with cosmological observations of stellar aberration is studied in
\cite[Chap.~13]{mybook03}
and
\cite[Sec.~10.2]{mybook05}.
The extension of the gyroparallelogram addition law of $k=2$ summands into
a higher dimensional gyroparallelotope addition law of $k>2$ summands
is mentioned in \eqref{eincoadd}\,--\,\eqref{ejsnk06}
and studied in \cite[Theorem 10.6]{mybook03}.

\section{The Isomorphism Between M\"obius and Einstein Addition} \label{secf2}

\begin{figure}[t]  
 \sidebyside {       
\psfrag{pa}  {$\hspace{0.1cm}A,\hspace{0.1cm}t=0$}
\psfrag{pb}  {$\hspace{0.1cm}B,\hspace{0.1cm}t=1$}
\psfrag{formula01}[]{${\rm The~M\ddot obius~Gyroline}$}
\psfrag{formula02}[]{${\rm through~the~points}~X~{\rm and}~B$}
\psfrag{formula03}[]{$\boxed{A\op (\om  A \op  B)\od  t}$}
\psfrag{fig174m04}[]{$ -\infty < t < \infty $}
\psfrag{formula05}[]{$\op=\opm$}
 \includegraphics[width=0.45\textwidth]{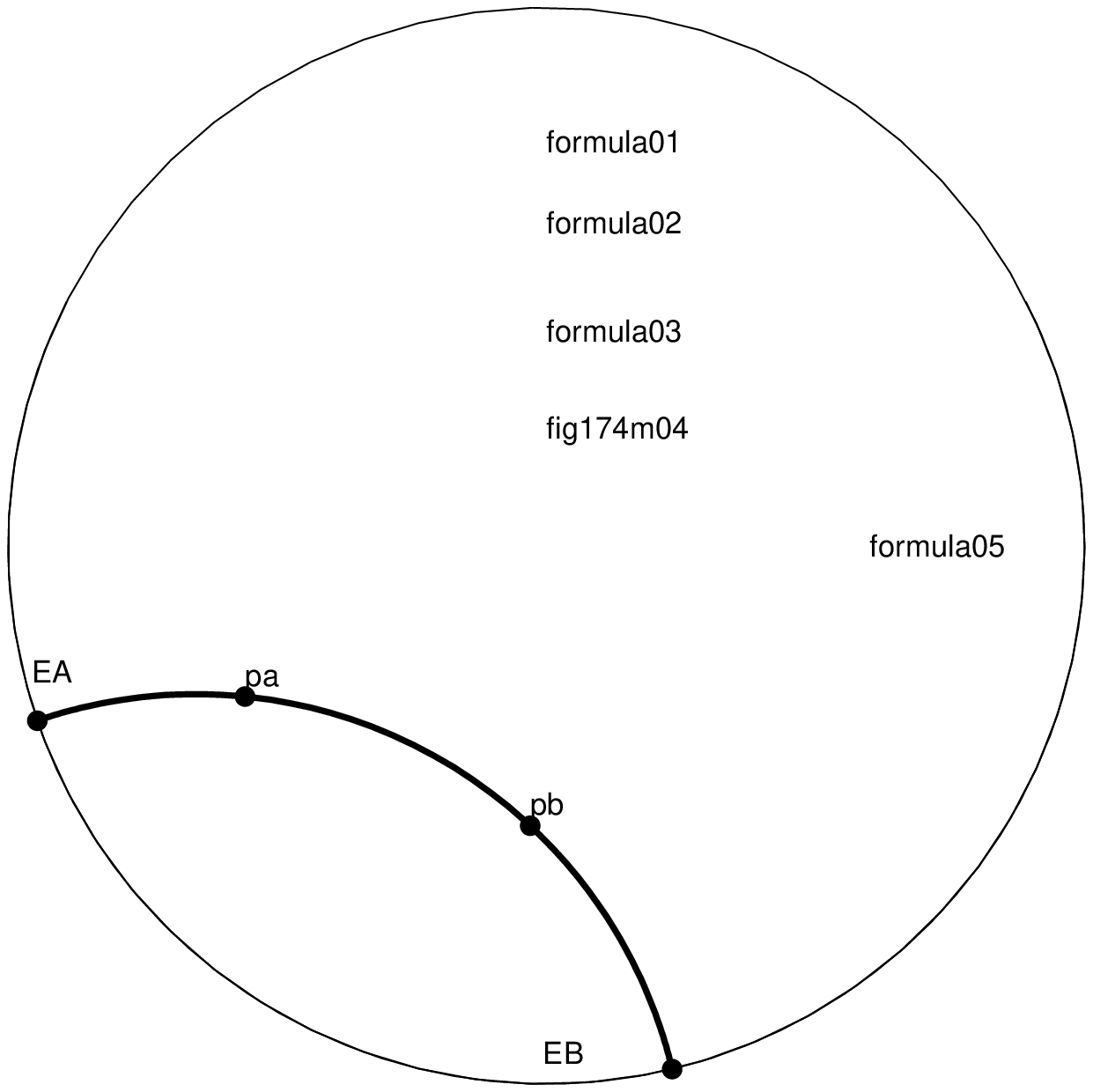}
\caption[The M\"obius Gyroline]{
The unique gyroline $L_{^{AB}}$ in a M\"obius gyrovector space $(\Rcn,\op ,\od )$
through two given points $A$ and $B$.
The case of the M\"obius gyrovector plane, when $\vc=\Rctwou$ is the
real open unit disc, is shown.
\label{fig174a1m}}}
  {
\psfrag{a}[]{$A$}
\psfrag{b}[]{$\!\! B$}
\psfrag{mab}[][][0.8]{$\hspace{2.8cm}M_{^{AB}}=\half\od(A\sqp B)$}
\psfrag{pab}{$P$}
\psfrag{formula00}[]{$d_{\op}(A,P) \op d_{\op}(P,B)=d_{\op}(A,B)$}
\psfrag{formula01}[]{$\boxed{A\op(\om A \op B)\od t}$}
\psfrag{formula02}[]{$0\le t \le1$}
\psfrag{formula05}[]{$\op=\opm$}
 \includegraphics[width=0.45\textwidth]{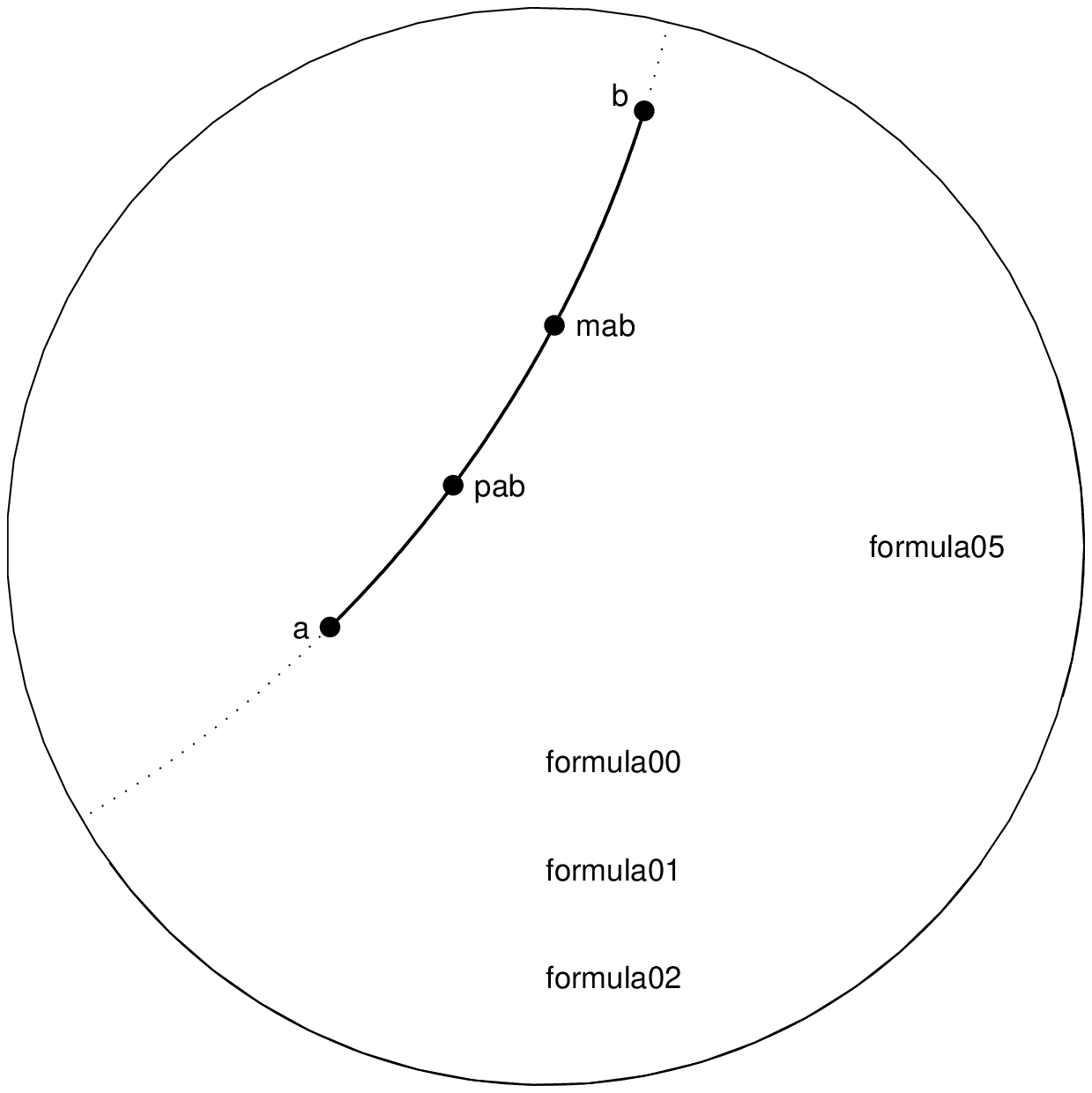}
\caption{
The gyrosegment $AB$ that links
the points $A$ and $B$ in $(\Rcn,\op ,\od )$, with one of its generic points
$P$ and its gyromidpoint $M_{^{AB}}$. The point $P$ lies between $A$ and $B$
and, hence, obeys the gyrotriangle equality, \eqref{rif01s}.
\label{fig163a1m}} }
\end{figure}


Einstein addition, $\op=\ope$, and M\"obius addition, $\opm$, admit the same
scalar multiplication \eqref{eqmlt03}, $\od=\ode=\odm$.
The isomorphism between $\ope$ and $\opm$ is given by the identities
 \begin{equation} \label{eq3ejksa}
 \begin{split}
 A \ope B  &= 2\od(\half\od A \opm\half\od B ),
\hspace{1.2cm} A,B \inn (\Rcn.\ope,\ode)
 \\[6pt]
 A \opm B  &= \half\od(2\od A \ope 2\od B ),
\hspace{1.18cm} A,B \inn (\Rcn.\opm,\odm)
 \end{split}
 \end{equation}
for all $A,B\in\Rcn$.

The isomorphism in \eqref{eq3ejksa} is not trivial owing to the result that
scalar multiplication, $\od$, is non-distributive, that is,
it does not distribute over gyrovector addition, $\op$.

As examples of the use of the isomorphism \eqref{eq3ejksa}
let
$A_e \inn (\Rcn,\ope,\od)$ and $A_m \inn (\Rcn,\opm,\od)$
be points of an Einstein and a M\"obius gyrovector space that are isomorphic
to each other under the isomorphism \eqref{eq3ejksa}. Then,
\begin{equation} \label{hdkeb1}
\begin{split}
A_e &= 2\od A_m \\[4pt]
A_m &= \half\od A_e
\end{split}
\end{equation}
It follows from \eqref{hdkeb1} that
\begin{equation} \label{hdkeb2}
\begin{split}
\gamma_{_{A_e}}^{\phantom{1}} &= \gamma_{_{2\od A_m}}^{\phantom{1}}
= 2 \gamma_{_{A_m}}^{2} - 1
\\[6pt]
\gamma_{_{A_e}}^{\phantom{1}} A_e &=
\gamma_{_{_{2\od A_m}}}^{\phantom{1}} (2\od A_m)
= 2 \gamma_{_{A_m}}^{2} A_m
\end{split}
\end{equation}

More generally, for points $A_{i,e}, A_{j,e} \inn (\Rsn,\ope,\od)$ in an
Einstein gyrovector space and their isomorphic image
$A_{i,m}, A_{j,m} \inn (\Rsn,\ope,\od)$ in a corresponding
M\"obius gyrovector space,
we have, \cite[Eq.~(2.278)]{mybook06},
\begin{equation} \label{hdkeb3}
\gamma_{ij,e}^{\phantom{1}} :=
\gamma_{_{\ome A_{i,e} \ope A_{j,e}}}^{\phantom{1}}
=
2\gamma_{_{\omm A_{i,m} \opm A_{j,m}}}^{2} -1
=:
2\gamma_{ij,m}^{2} - 1
\end{equation}
and, \cite[Eq.~(2.280)]{mybook06},
\begin{equation} \label{hdkeb3s}
\sqrt{ \gamma_{ij,e}^2 -1 }
=
2\gamma_{ij,m}^{\phantom{1}}
\sqrt{ \gamma_{ij,m}^2 -1 }
\end{equation}

Interestingly, in the following equation we see an elegant expression
that remains invariant under the isomorphism \eqref{eq3ejksa} between
Einstein and M\"obius gyrovector spaces:
\begin{equation} \label{gtres}
\frac{\gamma_{ij,e} \ab_{ij,e}}{\sqrt{ \gamma_{ij,e}^2 -1 }}
=
\frac{\gamma_{ij,m} \ab_{ij,m}}{\sqrt{ \gamma_{ij,m}^2 -1 }}
\end{equation}
as one can readily check,
where we use the notation
\begin{equation} \label{gtres1}
\begin{split}
\ab_{ij,e} &= \ome A_{i,e} \ope  A_{j,e}
\\
\ab_{ij,m} &= \omm A_{i,m} \opm  A_{j,m}
\\
\gamma_{ij,e}^{\phantom{O}} &= \gamma_{\ab_{ij,e}}^{\phantom{O}}
\\
\gamma_{ij,m}^{\phantom{O}} &= \gamma_{\ab_{ij,m}}^{\phantom{O}}
\end{split}
\end{equation}

A study in detail of
the isomorphism between Einstein and M\"obius gyrovector spaces
is found in
\cite[Sec.~6.21]{mybook03} and \cite[Sec.~2.29]{mybook06}.

Owing to the isomorphism between Einstein and M\"obius addition in $\Rcn$,
the triples $(\Rcn,\opm,\od)$ form M\"obius gyrovector spaces
just as
the triples $(\Rcn,\ope,\od)$ form Einstein gyrovector spaces.
We will now show in \eqref{ruhd1}--\eqref{ruhd3} below that the
isomorphic image of an Einstein gyroline $P_e(t)$ in an
Einstein gyrovector space $(\Rcn,\ope,\od)$ is
a M\"obius gyroline $P_m(t)$ in a corresponding
M\"obius gyrovector space $(\Rcn,\opm,\od)$.

Let
\begin{equation} \label{ruhd1}
P_e(t) = A_e\ope(\ome A_e\ope B_e)\od t
\end{equation}
$t\in\Rb$, be the gyroline that passes through the distinct points
$A_e,B_e\in\Rcn$ in an Einstein gyrovector space $(\Rcn,\ope,\od)$,
shown in Fig.~\ref{fig174b3enm}, p.~\pageref{fig174b3enm} for $n=2$.
Furthermore, let $A_m,B_m,P_m\in\Rcn$ be the respective isomorphic images
of the points $A_e,B_e,P_e\in\Rcn$ in \eqref{ruhd1} under the isomorphism
expressed in \eqref{eq3ejksa}--\eqref{hdkeb1}.
In the following chain of equations, which are numbered for subsequent explanation,
we determine the isomorphic image of the Einstein gyroline \eqref{ruhd1}
in the corresponding M\"obius gyrovector space $(\Rcn,\opm,\od)$.
\begin{equation} \label{ruhd2}
\begin{split}
2\od P_m(t)
&
\overbrace{=\!\!=\!\!=}^{(1)} \hspace{0.2cm}
2\od A_m\ope(\ome 2\od A_m\ope 2\od B_m)\od t
\\[4pt]&
\overbrace{=\!\!=\!\!=}^{(2)} \hspace{0.2cm}
2\od A_m\ope(2\od(-A_m) \ope 2\od B_m)\od t
\\[4pt]&
\overbrace{=\!\!=\!\!=}^{(3)} \hspace{0.2cm}
2\od A_m\ope2\od\{ (-A_m\opm B_m)\}\od t
\\[4pt]&
\overbrace{=\!\!=\!\!=}^{(4)} \hspace{0.2cm}
2\od A_m\ope2\od\{ (-A_m\opm B_m)\od t\}
\\[4pt]&
\overbrace{=\!\!=\!\!=}^{(5)} \hspace{0.2cm}
2\od\{ A_m \opm(-A_m\opm B_m)\od t\}
\\[4pt]&
\overbrace{=\!\!=\!\!=}^{(6)} \hspace{0.2cm}
2\od\{ A_m \opm(\omm A_m\opm B_m)\od t\}
\end{split}
\end{equation}
so that, finally, the two extreme sides of \eqref{ruhd2} give the equation
\begin{equation} \label{ruhd3}
P_m(t) = A_m \opm(\omm A_m\opm B_m)\od t
\end{equation}

Derivation of the numbered equalities in \eqref{ruhd2} follows:
\begin{enumerate}
\item\label{kapit1}
This equation follows from \eqref{ruhd1} and \eqref{hdkeb1}, where the equations
$P_e=2\od P_m$, $A_e=2\od A_m$ and $B_e=2\od B_m$
that result from \eqref{hdkeb1} are substituted into \eqref{ruhd1}.
\item\label{kapit2}
Follows from Item (1) since the unary operations $\ome$ and -- are identically
the same in Einstein gyrovector spaces, and since
$-2\od A_m=2\od(-A_m)$.
\item\label{kapit3}
Follows from Item (2) by the first identity in \eqref{eq3ejksa} applied to the
second binary operation $\ope$ in Item (2).
\item\label{kapit4}
Follows from Item (3) by the scalar associative law of gyrovector spaces.
\item\label{kapit5}
Follows from Item (4) by the first identity in \eqref{eq3ejksa} applied to the
remaining binary operation $\ope$ in Item (4).
\item\label{kapit6}
Follows from Item (5) since the unary operations $\omm$ and -- are identically
the same in M\"obius gyrovector spaces.
\end{enumerate}

A M\"obius gyroline in a M\"obius gyrovector plane $(\Rctwo,\op,\od)$ is
shown in Fig.~\ref{fig174a1m}.
Interestingly, a M\"obius gyroline that does not pass through the
center of the disc $\Rctwo$ is a circular arc that approaches the boundary
of the disc orthogonally. This feature of the M\"obius gyroline indicates
that M\"obius gyrovector spaces form the algebraic setting for the
Poincar\'e ball model of hyperbolic geometry.
The link between Einstein and M\"obius gyrovector spaces and
differential geometry is presented in \cite{ungardiff05}.

As in \eqref{tkcflsen}--\eqref{rif01}, but now with $\op=\opm$,
M\"obius addition $\op$ admits the gyrodistance function
\begin{equation} \label{tkcflsmb}
d_\op( A , B ) = \|\om A \op B \|
\end{equation}
that obeys the gyrotriangle inequality
\cite[Theorem 6.9]{mybook03}
\begin{equation} \label{rif01t}
d_\op( A , B ) \le d_\op( A , P ) \op d_\op( P , B )
\end{equation}
for any $ A , B , P \in\Rcn$ in a M\"obius gyrovector space $(\Rcn,\op,\od)$.
M\"obius gyrodistance function is invariant under the group of motions of its
M\"obius gyrovector space, that is, under left gyrotranslations and
rotations of the space \cite[Sec.~4]{mybook03}. The gyrotriangle inequality \eqref{rif01t}
reduces to a corresponding gyrotriangle equality,
\begin{equation} \label{rif01ts}
d_\op( A , B ) = d_\op( A , P ) \op d_\op( P , B )
\end{equation}
if and only if point $ P $ lies between points $ A $ and $ B $, that is,
point $ P $ lies on the gyrosegment $ A  B $, as shown in
Fig.~\ref{fig163a1m}.
Accordingly, the gyrodistance function is gyroadditive on gyrolines, as
demonstrated in \eqref{rif01ts} and
illustrated graphically in Fig.~\ref{fig163a1m}.

The one-to-one relationship between M\"obius gyrodistance function
\eqref{tkcflsmb} and the standard Poincar\'e distance function
in the Poincar\'e ball model of hyperbolic geometry is presented in
\cite[Sec.~6.17]{mybook03}.

Einstein coaddition, $\sqp=\sqpe$, in the ball, defined in \eqref{eqdfhn01},
is commutative as shown in \eqref{eincoadd}. Its importance stems from
analogies with classical results that it captures.
In particular, it proves useful in solving the gyrogroup equation
\eqref{eqtghj03}, in the determination of gyromidpoints in
\eqref{eqgyromid}, and in the formulation of the
gyroparallelogram addition law in \eqref{ryuf1} and in Fig.~\ref{fig190k4m}.

\section{M\"obius Coaddition}\label{secc4s}

We now wish to determine M\"obius coaddition in the ball $\Rcn$
by means of the isomorphism between M\"obius and Einstein gyrovector spaces.
Let $\ub_e,\vb_e,\wb_e\in(\Rcn,\ope,\od)$ be three elements of
an Einstein gyrovector space such that
\begin{equation} \label{rkdt}
\wb_e=\ub_e\sqpe\vb_e
\end{equation}
and
let $\ub_m,\vb_m,\wb_m\in(\Rcn,\opm,\od)$ be the corresponding elements of
of the corresponding M\"obius gyrovector space.
Then,
\begin{equation} \label{rkds}
\wb_m=\ub_m\sqpm\vb_m
\end{equation}
where M\"obius coaddition $\sqpm$ in $(\Rcn,\opm,\od)$ is determined from
Einstein coaddition $\sqpe$ in
the following chain of equations, which are numbered for
subsequent explanation.
\begin{equation} \label{grrdkw}
\begin{split}
\ub_m \sqpm \vb_m
&
\overbrace{=\!\!=\!\!=}^{(1)} \hspace{0.2cm}
\wb_m
\\[4pt]&
\overbrace{=\!\!=\!\!=}^{(2)} \hspace{0.2cm}
\half\od\wb_e
\\[4pt]&
\overbrace{=\!\!=\!\!=}^{(3)} \hspace{0.2cm}
\half\od(\ub_e\sqpe\vb_e)
\\[4pt]&
\overbrace{=\!\!=\!\!=}^{(4)} \hspace{0.2cm}
\half\od\left\{
2\od\frac{
\gamma_{_{\ub_e}}^{\phantom{1}} \ub_e
+
\gamma_{_{\vb_e}}^{\phantom{1}} \vb_e
}{
\gamma_{_{\ub_e}}^{\phantom{1}}
+
\gamma_{_{\vb_e}}^{\phantom{1}}
}
\right\}
\\[4pt]&
\overbrace{=\!\!=\!\!=}^{(5)} \hspace{0.2cm}
\frac{
\gamma_{_{\ub_e}}^{\phantom{1}} \ub_e
+
\gamma_{_{\vb_e}}^{\phantom{1}} \vb_e
}{
\gamma_{_{\ub_e}}^{\phantom{1}}
+
\gamma_{_{\vb_e}}^{\phantom{1}}
}
\\[4pt]&
\overbrace{=\!\!=\!\!=}^{(6)} \hspace{0.2cm}
\frac{
2 \gamma_{_{\ub_m}}^{2} \ub_m
+
2 \gamma_{_{\vb_m}}^{2} \ub_m
}{
2 \gamma_{_{\ub_m}}^{2} - 1
+
2 \gamma_{_{\vb_m}}^{2} - 1
}
\\[4pt]&
\overbrace{=\!\!=\!\!=}^{(7)} \hspace{0.2cm}
\frac{
\gamma_{_{\ub_m}}^{2} \ub_m
+
\gamma_{_{\vb_m}}^{2} \ub_m
}{
\gamma_{_{\ub_m}}^{2}
+
\gamma_{_{\vb_m}}^{2} - 1
}
\end{split}
\end{equation}

Derivation of the numbered equalities in \eqref{grrdkw} follows:
\begin{enumerate}
\item\label{kanut1}
The equation in Item (1) is \eqref{rkds}.
\item\label{kanut2}
The equation in Item (2) follows from the isomorphism \eqref{hdkeb1}
between $\wb_m$ in a M\"obius gyrovector space $(\Rcn,\opm,\od)$
and its isomorphic image $\wb_e$ in the isomorphic
Einstein gyrovector space $(\Rcn,\ope,\od)$.
\item\label{kanut3}
Follows from (2) by assumption \eqref{rkdt}.
\item\label{kanut4}
Follows from (3) by \eqref{eincoadd}.
\item\label{kanut5}
Follows from (4) by the scalar associative law of gyrovector spaces,
Sec.~\ref{secc4}.
\item\label{kanut6}
Follows from (5) by \eqref{hdkeb2}.
\end{enumerate}

Hence, by \eqref{grrdkw}, M\"obius coaddition $\sqpm$ in a M\"obius
gyrovector space $(\Rcn,\opm,\od)$ is given by the equation
\begin{equation} \label{kesin}
\ub\sqpm\vb = \frac{
\gamma_{_{\ub}}^{2} \ub
+
\gamma_{_{\vb}}^{2} \ub
}{
\gamma_{_{\ub}}^{2}
+
\gamma_{_{\vb}}^{2} - 1
}
\end{equation}
for all $\ub,\vb\in\Rcn$.

In order to extend \eqref{grrdkw} from M\"obius coaddition of order two
to any order $k$, $k>2$, we rewrite \eqref{ejsnk06} in the form
\begin{equation} \label{mishne01}
\wb_e :=
\vb_{1,e} \sqp_{E,k} \vb_{2,e} \sqp_{E,k} \, \ldots \, \sqp_{E,k} \vb_{k,e}
= 2\od \frac{\sum_{i=1}^{k}\gvbie\vb_{i,e}}{2+\sum_{i=1}^{k}(\gvbie -1)}
\end{equation}
where $\vb_{i,e} \in (\Rcn,\ope,\od)$, $i=1,\ldots,k$, are $k$ elements
of an Einstein gyrovector space and where $\wb_e\in (\Rcn,\ope,\od)$
is their {\it cosum}, $\sqp_{E,k}$ being the Einstein $k$-ary cooperation,
that is, the Einstein cooperation of order $k$.

Let $\vb_{i,m}$, $i=1,\ldots,k$, and $\wb_m$ be the respective
isomorphic images of $\vb_{i,e}$, and $\wb_e$
in the corresponding M\"obius gyrovector space $(\Rcn,\opm,\od)$,
under isomorphism \eqref{hdkeb1}. Then,
\begin{equation} \label{mishne02}
\wb_m =
\vb_{1,m} \sqp_{M,k} \vb_{2,m} \sqp_{M,k} \, \ldots \, \sqp_{M,k} \vb_{k,m}
\end{equation}
where M\"obius coaddition of order $k$, $\sqp_{M,k}$,
is to be determined in the chain of equations below,
which are numbered for subsequent interpretation:
\begin{equation} \label{mishne03}
\begin{split}
\vb_{1,m} \sqp_{M,k} & \vb_{2,m} \sqp_{M,k} \, \ldots \, \sqp_{M,k} \vb_{k,m}
\overbrace{=\!\!=\!\!=}^{(1)} \hspace{0.2cm}
\wb_m
\\[4pt]&
\overbrace{=\!\!=\!\!=}^{(2)} \hspace{0.2cm}
\half\od\wb_e
\\[4pt]&
\overbrace{=\!\!=\!\!=}^{(3)} \hspace{0.2cm}
\half\od(\vb_{1,e} \sqp_{E,k} \vb_{2,e} \sqp_{E,k} \, \ldots \, \sqp_{E,k} \vb_{k,e})
\\[4pt]&
\overbrace{=\!\!=\!\!=}^{(4)} \hspace{0.2cm}
\half\od\left\{
2\od \frac{\sum_{i=1}^{k}\gvbie\vb_{i,e}}{2+\sum_{i=1}^{k}(\gvbie -1)}
\right\}
\\[4pt]&
\overbrace{=\!\!=\!\!=}^{(5)} \hspace{0.2cm}
\frac{\sum_{i=1}^{k}\gvbie\vb_{i,e}}{2+\sum_{i=1}^{k}(\gvbie -1)}
\\[4pt]&
\overbrace{=\!\!=\!\!=}^{(6)} \hspace{0.2cm}
\frac{2\sum_{i=1}^{k} \gamma_{\vb_i,m}^2 \vb_{i,m}
}{
2+\sum_{i=1}^{k}(2\gamma_{\vb_i,m}^2-2)
}
\\[4pt]&
          {=\!\!=\!\!=}^{~~~} \hspace{0.2cm}
\frac{ \sum_{i=1}^{k} \gamma_{\vb_i,m}^2 \vb_{i,m}
}{
1+\sum_{i=1}^{k}(\gamma_{\vb_i,m}^2-1)
}
\end{split}
\end{equation}

Derivation of the numbered equalities in \eqref{grrdkw} follows:
\begin{enumerate}
\item\label{komot1}
The equation in Item (1) is \eqref{mishne02}.
\item\label{komot2}
The equation in Item (2) follows from the isomorphism \eqref{hdkeb1}
between $\wb_m$ in a M\"obius gyrovector space $(\Rcn,\opm,\od)$
and its isomorphic image $\wb_e$ in the isomorphic
Einstein gyrovector space $(\Rcn,\ope,\od)$.
\item\label{komot3}
Follows from (2) by the assumption in \eqref{mishne01}.
\item\label{komot4}
Follows from (3) by the equation in \eqref{mishne01}.
\item\label{komot5}
Follows from (4) by the scalar associative law of gyrovector spaces,
Sec.~\ref{secc4}.
\item\label{komot6}
Follows from (5) by \eqref{hdkeb2}.
\end{enumerate}

Hence, by \eqref{mishne03},
M\"obius coaddition of order $k$, $\sqp_{M,k}$ in a M\"obius
gyrovector space $(\Rcn,\opm,\od)$ is given by the equation
\begin{equation} \label{mishne04}
\vb_{1,m} \sqp_{M,k} \vb_{2,m} \sqp_{M,k} \, \ldots \, \sqp_{M,k} \vb_{k,m}
=
\frac{ \sum_{i=1}^{k} \gamma_{\vb_i,m}^2 \vb_{i,m}
}{
1+\sum_{i=1}^{k}(\gamma_{\vb_i,m}^2-1)
}
\end{equation}
for all $\vb_{i,m}\in(\Rcn,\opm,\od)$, $i=1,\ldots,k$.

\section{M\"obius double-gyroline} \label{secg2a}

\begin{theorem}\label{mudyn}
Let $A,B\in\Rcn$ be any two distinct points of a
M\"obius gyrovector space $(\Rcn,\op,\od)$, and let
\begin{equation} \label{kromd01}
\LAB(t) = A\op(\om A\op B)\od t
\end{equation}
$t\in\Rb$, be the gyroline that passes through these points.
Then,
\begin{equation} \label{kromd02}
2\od\LAB(t) = A\sqp\LAB(2t)
\end{equation}
\end{theorem}
\begin{proof}
Let
\begin{equation} \label{kromd03}
\begin{split}
F_1(t) &= (\om A\op B)\od t \\
F_2(t) &= 2\od F_1(t)
\end{split}
\end{equation}
so that we have, by the scalar associative law of gyrovector spaces,
\begin{equation} \label{kromd04}
\begin{split}
F_2(t) &= 2\od F_1(t) \\
&= 2\od(\om A\op B)\od t \\
&= (\om A\op B)\od(2t) \\
&= F_1(2t)
\end{split}
\end{equation}

Hence, by \eqref{kromd03}\,--\,\eqref{kromd04},
\eqref{kromd02} can be written equivalently as
\begin{equation} \label{kromd05}
2\od(A\op F_1(t)) = A\sqp(A\op F_1(2t)) = A\sqp(A\op F_2(t))
\end{equation}
so that instead of verifying \eqref{kromd02} we can, equivalently,
verify \eqref{kromd05}.

The proof of  \eqref{kromd05} is given by
the following chain of equations, which are numbered for subsequent
derivation.

\begin{equation} \label{kromd06}
\begin{split}
2\od(A\op F_1)
&
\overbrace{=\!\!=\!\!=}^{(1)} \hspace{0.2cm}
A\op(2\od F_1\op A)
\\[2pt]&
\overbrace{=\!\!=\!\!=}^{(2)} \hspace{0.2cm}
A\op(F_2\op A)
\\[2pt]&
\overbrace{=\!\!=\!\!=}^{(3)} \hspace{0.2cm}
A\sqp(A\op F_2)
\end{split}
\end{equation}
as desired.

Derivation of the numbered equalities in \eqref{kromd06} follows:
\begin{enumerate}
\item\label{katon1}
Follows from the Two-Sum Identity \eqref{hyets1}.
\item\label{katon2}
Follows from \eqref{kromd03}.
\item\label{katon3}
Follows from \eqref{fromu1}.
\end{enumerate}
 
\end{proof}

  
\begin{figure}[t]  
 \centering         
\psfrag{A}{$A$}
\psfrag{B}{$B$}
\psfrag{A2}{$2\od A$}
\psfrag{B2}{$2\od B$}
\psfrag{ApLAB2t}[]{$A\sqp\LAB(2t)$}
\psfrag{TLABt}[]{$2\od\LAB(t)$}
\psfrag{pah}[]{$\half\od A$}
\psfrag{pbh}[]{$B\om \half\od A$}
\psfrag{LABt}[]{$\LAB(t)$}
\psfrag{EA}[]{$E_{^A}$}
\psfrag{EB}[]{$E_{^B}$}
\psfrag{formula01}[]{$L_{^{\half\od A,B\om\half\od A}}(t)$}
\psfrag{formula02}[]{$-\infty<t<\infty$}
\psfrag{formula01}[]{$\boxed{2\od\LAB(t)=A\sqp\LAB(2t)}$}
 \includegraphics[width=9cm]{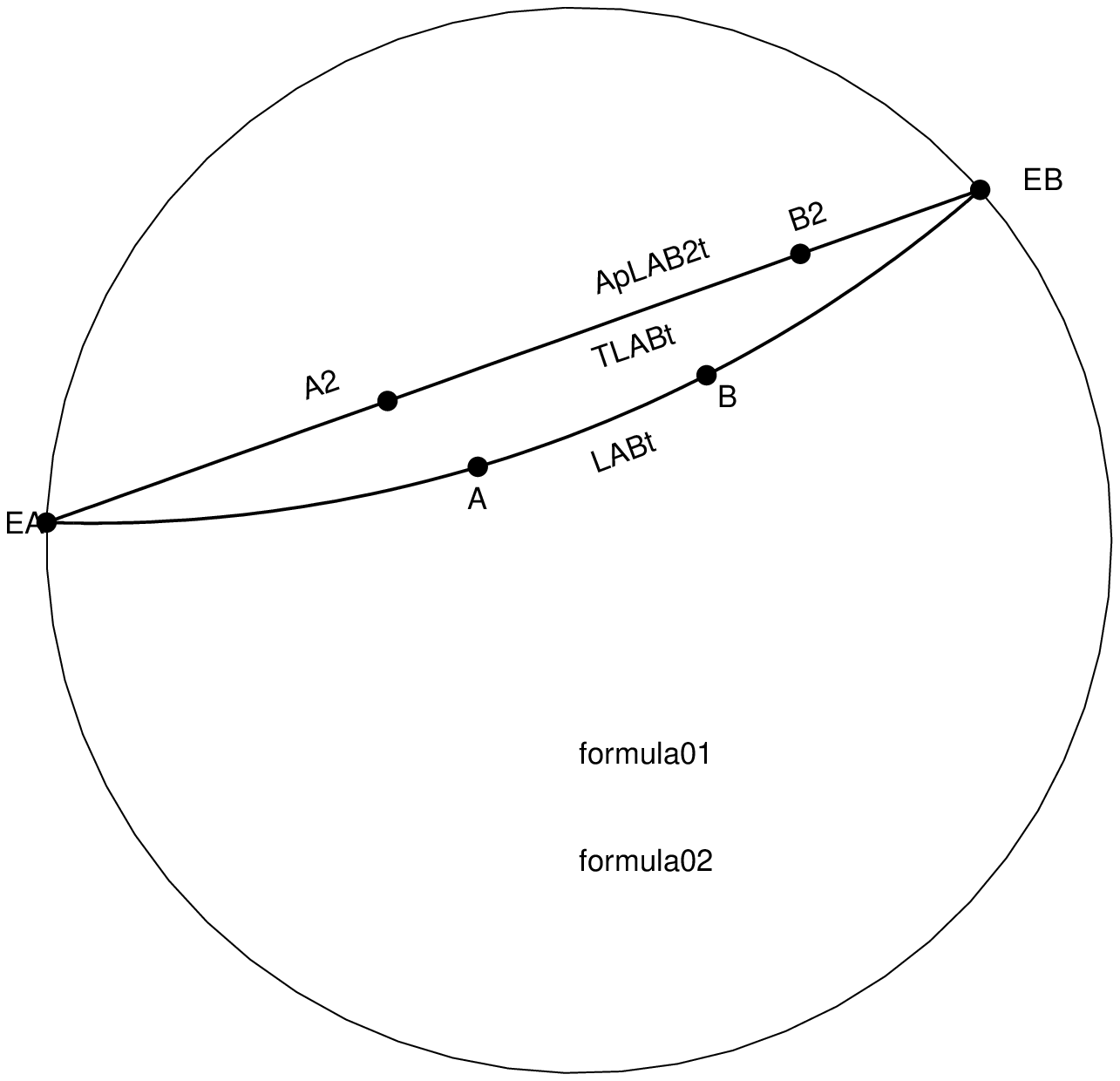}
\caption[]{
$A$ and $B$ are any two given distinct points of a M\"obius gyrovector space $(\Rcn,\op,\od)$.
The gyroline that passes through the points $A,B\in\Rcn$ is $\LAB(t)$, $-\infty<t<\infty$,
and its corresponding double-gyroline is $2\od\LAB(t)$, so that is passes through the points
$2\od A,2\od B\in\Rcn$. The latter turns out to be the
Euclidean straight line in the ball that passes through the points
$2\od A$ and $2\od B$. Furthermore, the double-gyroline $2\od\LAB(t)$, parametrized by $t$,
is identical with the cogyrotranslation by $A$,
$A\sqp\LAB(2t)$, of its gyroline, parametrized by $2t$,
as shown here for $n=2$.
\label{fig303m}}
\end{figure}


We may remark that in the Euclidean limit, when the radius $c$ of the
ball $\Rcn$ tends to $\infty$, the ball expands to the whole of its
Euclidean $n$-space $\Rn$, both M\"obius addition $\op$ and
coaddition $\sqp$ in the ball $\Rcn$ reduce to the common vector addition +
in the space $\Rn$, and Identity \eqref{kromd02} of Theorem \ref{mudyn}
in the ball $\Rcn$ reduces to the trivial identity in $\Rn$,
\begin{equation} \label{knomd}
2[A+(-A+B)t] = A+[A+(-A+B)2t]
\end{equation}
Thus, we see once again that in order to capture analogies with
classical results, both gyrogroup operation and cooperation
must be considered.

Theorem \ref{mudyn} suggests the following definition:
\begin{definition}\label{defdmbgyr}
{\bf (M\"obius double-gyroline).}
{\it
Let $A,B\in\Rcn$ be two distinct points of a
M\"obius gyrovector space $(\Rcn,\op,\od)$, and let
\begin{equation} \label{kromd01s}
\LAB(t) = A\op(\om A\op B)\od t
\end{equation}
$t\in\Rb$, be the gyroline that passes through these points.
The M\"obius double-gyroline $\PAB(t)$ of gyroline $\LAB(t)$ is
the curve given by the equation
\begin{equation} \label{kromd02s}
\PAB(t) = 2\od\LAB(t)
\end{equation}
$t\in\Rb$.
}
\end{definition}

Following Def.~\ref{defdmbgyr}, the gyrovector space identity
\eqref{kromd02} of Theorem \ref{mudyn} states that the
double-gyroline of a given gyroline that passes through a point $A$
in a M\"obius gyrovector space coincides with the
cogyrotranslation by $A$ of the gyroline.

Remarkably, the double-gyroline of a given gyroline in a
M\"obius gyrovector space turns out to be the supporting chord
of the gyroline, as shown in Fig.~\ref{fig303m}
and studied in Sec.~\ref{secg2b}.

Identity \eqref{kromd02} of Theorem \ref{mudyn} can be written, equivalently, as
\begin{equation} \label{gfne1}
\LAB(t) = \half\od(A\sqp\LAB(2t))
\end{equation}
Let $P(t)$ be a generic point on a gyroline $\LAB(t)$ for some value $t$
of the gyroline parameter $t$, so that $P(0)=A$ and $P(1)=A$.
Then, \eqref{gfne1} implies the equation
\begin{equation} \label{gfne2}
P(t) = \half\od P(0)\sqp P(2t)
\end{equation}
$t\in\Rb$. Equation \eqref{gfne2}, in turn, demonstrates that any point $P(t)$
of a gyroline $\LAB(t)$ is the midpoint of the points $P(0)=A$ and $P(2t)$ of the
gyroline, as explained in \eqref{eqgyromid}, p.~\pageref{eqgyromid}.

\section{Euclidean straight lines in M\"obius gyrovector spaces} \label{secg2b}

Euclidean straight lines (lines, in short)
appear naturally in Einstein gyrovector space balls
where they form gyrolines, as shown in Fig.~\ref{fig174b3enm}.
In this section we employ the isomorphism \eqref{eq3ejksa}
between Einstein and M\"obius gyrovector spaces for the task of expressing
lines in M\"obius gyrovector spaces.

Let $A,B\in(\Rcn,\opm,\od)$ be two distinct points of a
M\"obius gyrovector space. We know that the unique gyroline in an
Einstein gyrovector spaces $(\Rcn,\ope,\od)$ that passes through
the points $A,B\in\Rcn$ is the set of point $\PAB(t)$ given by
\begin{equation} \label{frut01}
\PAB(t) = A\ope(\ome A\ope B)\od t
\end{equation}
$t\in\Rb$. It is the intersection of a line and the ball $\Rcn$,
as shown in Fig.~\ref{fig174b3enm} for the disc $\Rctwo$.
This line passes through the point $A$ when $t=0$
and through the point $B$ when $t=1$.

Unlike Einstein gyrolines,
which are line segments, M\"obius gyrolines are Euclidean circular arcs
that intersect the boundary of the ball $\Rcn$ orthogonally,
as shown in Fig.~\ref{fig174a1m} for the disc $\Rctwo$.
In order to accomplish the task we face in this section, in
the following chain of equations \eqref{frut02} we express
\eqref{frut01} in terms of M\"obius addition $\opm$ rather than
Einstein addition $\ope$, noting that both
Einstein and M\"obius scalar multiplication $\od$ are identically the same,
as remarked in Sec.~\ref{secf2}.
Starting from \eqref{frut01}, we have the following
chain of equations, which are numbered for subsequent derivation:

\begin{equation} \label{frut02}
\begin{split}
\PAB(t) & {=\!\!=\!\!=} \hspace{0.2cm}
A\ope(\ome A\ope B)\od t
\\&
\overbrace{=\!\!=\!\!=}^{(1)} \hspace{0.2cm}
A\ope((-A)\ope B)\od t
\\&
\overbrace{=\!\!=\!\!=}^{(2)} \hspace{0.2cm}
2\od\{\half\od A \opm\half\od((-A)\ope B)\od t\}
 \\&
 \overbrace{=\!\!=\!\!=}^{(3)} \hspace{0.2cm}
\half\od A \opm\{((-A)\ope B)\od t\opm \half\od A\}
\\&
\overbrace{=\!\!=\!\!=}^{(4)} \hspace{0.2cm}
\half\od A \opm\{2\od((-\half\od A)\opm\half B)\od t\opm\half\od A\}
\\&
\overbrace{=\!\!=\!\!=}^{(5)} \hspace{0.2cm}
\half\od A \opm\{[(-\half\od A)\opm(B\opm(-\half\od A))]\od t \opm \half\od A\}
\\&
\overbrace{=\!\!=\!\!=}^{(6)} \hspace{0.2cm}
\half\od A \sqpm\{\half\od A\opm[(-\half\od A) \opm(B\opm(-\half\od A))]\od t\}
\\&
\overbrace{=\!\!=\!\!=}^{(7)} \hspace{0.2cm}
\half\od A \sqpm\{\half\od A\opm[\omm\half\od A \opm(B\omm\half\od A)]\od t\}
\end{split}
\end{equation}

Hence, by \eqref{frut02},
\begin{equation} \label{furna}
\PAB(t) =
\half\od A \sqpm\{\half\od A\opm[\omm\half\od A \opm(B\omm\half\od A)]\od t\}
\end{equation}
Derivation of the numbered equalities in \eqref{frut02} follows:
\begin{enumerate}
\item \label{uekif1}
Follows from the result that $\ome A=-A$ (as well as $\omm A=-A$; see
Item \ref{uekif7} below).
\item \label{uekif2}
Follows from isomorphism \eqref{eq3ejksa} between $\ope$ and $\opm$,
applying the isomorphism to the first $\ope$ in (1).
\item \label{uekif3}
Follows from the Two-Sum Identity, \eqref{hyets1}.
\item \label{uekif4}
Again, follows from isomorphism \eqref{eq3ejksa} between $\ope$ and $\opm$,
as in Item \ref{uekif2}, now
applying the isomorphism to the remaining $\ope$ in (3).
\item \label{uekif5}
Again, follows from the gyrogroup {\it Two-Sum Identity}, as in Item \ref{uekif3}.
\item \label{uekif6}
Follows from the gyrogroup identity \eqref{fromu1},
\begin{equation} \label{hyets2}
A\op(B\op A) = A\sqp (A\op B)
\end{equation}
\item \label{uekif7}
Follows from the result that $\omm A=-A$ (as well as $\ome A=-A$; see
Item \ref{uekif1} above).
\end{enumerate}

In both \eqref{frut01} and \eqref{furna}, the set of points $\PAB(t)$, $t\in\Rb$,
forms a line in the ball $\Rcn$ of the M\"obius gyrovector space
$(\Rcn,\opm,\od)$, where the points $A$ and $B$ lie.
In \eqref{frut01} this line is expressed in terms of operations of
Einstein gyrovector spaces
while in \eqref{furna} this line is expressed in terms of operations of
M\"obius gyrovector spaces, obtained by means of isomorphism \eqref{eq3ejksa} between
Einstein and M\"obius gyrovector spaces.
By \eqref{furna} we have the following theorem:

\begin{theorem}\label{khdkn}
Let $A$ and $B$ be two distinct points in a M\"obius gyrovector space
$(\Rcn,\opm,\od)$. The unique line that passes through these points,
Fig.~\ref{fig300m}, is given by the equation
\begin{equation} \label{frut03}
\PAB(t) =
\half\od A \sqpm\{\half\od A\opm[\omm\half\od A \opm(B\omm\half\od A)]\od t\}
\end{equation}
\end{theorem}

Let $A,B\in\Rcn$ be any two distinct points in a M\"obius gyrovector space
$(\Rcn,\opm,\od)$, and let
$L_{^{\half\od A,B\om\half\od A}}(t)$, $t\in\Rb$,
be the unique gyroline through the points
$\half\od A$ and $B\om \half\od A$.
Then, as shown in Fig.~\ref{fig174a1m}, the gyroline is given by the equation
\begin{equation} \label{fuly01}
L_{^{\half\od A,B\om\half\od A}}(t) =
\half\od A\op[\om\half\od A\op(B\om\half\od A)]\od t
\end{equation}
so that \eqref{frut03} can be written as
\begin{equation} \label{fuly02}
\PAB(t) = \half\od A \sqp
L_{^{\half\od A,B\om\half\od A}}(t)
\end{equation}

The line $\PAB(t)$ of Theorem \ref{khdkn} in \eqref{frut03}
is recognized by means of \eqref{fuly01}\,--\,\eqref{fuly02} as the
{\it cogyrotranslation} by $\half\od A$ of the
M\"obius gyroline \eqref{fuly01} that passes through the points
$\half\od A$ and $B\om\half\od A$.

As shown in Fig.~\ref{fig301m}, the line $\PAB(t)$ is the
{\it supporting chord} of the gyroline
$L_{^{\half\od A,B\om\half\od A}}(t)$.

Let
\begin{equation} \label{fuly03}
\begin{split}
C &= \half\od A \\
D &= B\om\half A
\end{split}
\end{equation}

Then, by the scalar associative law of gyrovector spaces and by the
right cancellation law \eqref{eq01bbb}, we have
\begin{equation} \label{fuly04}
\begin{split}
A &= 2\od C \\
B &= D\sqp\half\od A = D\sqp C = C\sqp D
\end{split}
\end{equation}
so that \eqref{fuly02} can be written as
\begin{equation} \label{fuly05}
P_{^{2\od C,C\sqp D}}^{\phantom{o}} =
C\sqp L_{^{CD}}^{\phantom{o}} (t)
\end{equation}
thus leading to the following theorem:

\begin{theorem}\label{khdkns}
Let $C,D\in\Rcn$ be two distinct points in a M\"obius gyrovector space
$(\Rcn,\opm,\od)$, and let
\begin{equation} \label{fuly06}
L_{^{CD}}^{\phantom{o}} (t) = C\op(\om C\op D)\od t
\end{equation}
$t\in\Rb$, be the gyroline that passes through the points $C$ and $D$.
Then, the supporting chord of gyroline $L_{^{CD}}^{\phantom{o}} (t)$
is the line given by the cogyrotranslation of the gyroline by $C$,
\begin{equation} \label{fuly07}
C\sqp L_{^{CD}}^{\phantom{o}} (t)
\end{equation}

Furthermore, the supporting chord passes through the points
$P_1,P_2,P_3$, Fig.~\ref{fig302m}, where
\begin{equation} \label{fuly08}
\begin{split}
P_1 &= C \sqp C = 2\od C \\
P_2 &= D \sqp D = 2\od D \\
P_3 &= C \sqp D
\end{split}
\end{equation}
\end{theorem}

Let $Q=\omm A\opm B$ so that, by the gyrogroup left cancellation law, \eqref{eq01b},
$A\opm Q=B$. Restricting the line parameter $t\in\Rb$ to $t\ge0$, we obtain the
Euclidean ray (ray, in short) $\PAB(t), t\ge0$. It is the ray with edge $A$ that
contains the point $A\opm Q=B$, which is the right gyrotranslation of $A$
by $Q$. As such, it contains the sequence of all successive
right gyrotranslation of $A$ by $Q$, that is,
the sequence
$P_0=A$, $P_1=A\opm Q$, $P_2=(A\opm Q)\opm Q$, $P_3=((A\opm Q)\opm Q)\opm Q$, {\it etc,}
as shown in Fig.~\ref{fig300m}.

  
\begin{figure}[t]  
 \centering         
 \psfrag{C}{$C$}
 \psfrag{D}{$D$}
 \psfrag{P1}[]{$P_{1}$}
 \psfrag{P2}[]{$P_{2}$}
 \psfrag{P3}[]{$P_{3}$}
 \psfrag{EC}[]{$E_{^C}$}
 \psfrag{ED}[]{$E_{^D}$}
 \psfrag{formula01}[]{$P_1=2\od C=C\sqp C$}
 \psfrag{formula02}[]{$P_2=2\od D=D\sqp D$}
 \psfrag{formula03}[]{$P_3=C\sqp D$}
 \includegraphics[width=9cm]{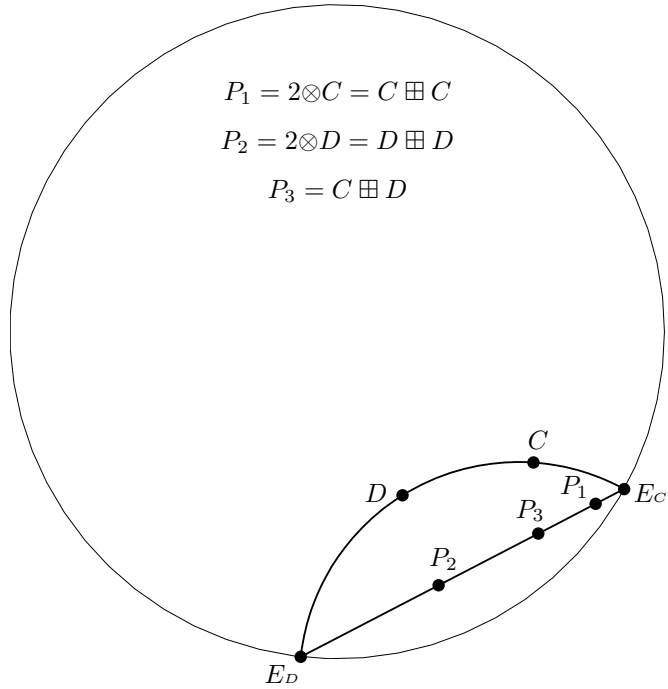}
\caption[]{
The Euclidean line $\PAB(t),~-\infty<t<\infty$, that passes through
the point $A$ and $B$
in a M\"obius gyrovector plane $(\Rstwo,\op,\od)$
is shown along with ...
TO BE COMPLETED
\label{fig302m}}
\end{figure}


  
\begin{figure}[t]  
 \centering         
\psfrag{A}{$A$}
\psfrag{B}{$B$}
\psfrag{90}{$P_0=\PAB(0)=A$}
\psfrag{91}{$P_1=\PAB(1)=B$}
\psfrag{92}{$P_2$}
\psfrag{93}{$P_3$}
\psfrag{94}{$P_4$}
\psfrag{95}{$P_5$}
 \includegraphics[width=9cm]{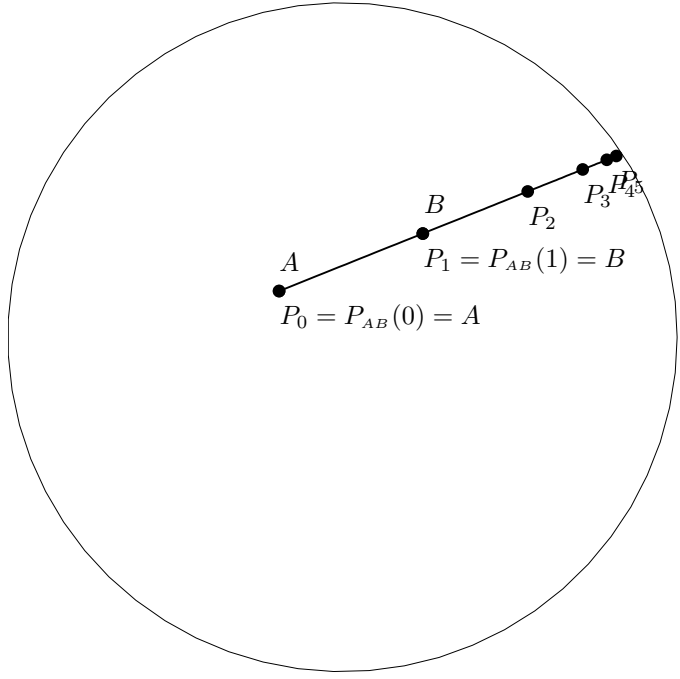}
 \caption[Mobius Gyroparallelogram Addition Law]{
The Euclidean ray $\PAB(t),~t\ge0$, with edge $A$ that passes through $B$
in a M\"obius gyrovector plane $(\Rstwo,\opm,\od)$
is shown along with several points of the sequence
of successive right gyrotranslations of $A$ by $Q=\om A\opm B$
that lies on the ray
(see also Fig.~6.4 in \cite[p.~168]{mybook01}).
\label{fig300m}}
\end{figure}


The points $P_k$, $k=1,2,3,\dots$ lie on the ray $\PAB(t)$, $t\ge0$,
as shown in Fig.~\ref{fig300m}, and as observed in
\cite[Figs.~6.3-6.5]{mybook01}.
Owing to the left loop property of gyrations in gyrogroup Axiom $(G5)$,
we have
\begin{equation} \label{frut04}
\gyr[\omm A \opm B, P_k] = \gyr[\omm A ,B]
\end{equation}
for all $k=0,1,2,3,\dots$. As an example,
the proof of \eqref{frut04} for $k=0$ and $k=1$ follows:

  
\begin{figure}[t]  
 \centering         
\psfrag{A}{$A$}
\psfrag{B}{$B$}
\psfrag{pah}[]{$\half\od A$}
\psfrag{pbh}[]{$B\om \half\od A$}
\psfrag{PAB}[]{$\PAB(t)$}
\psfrag{EA}[]{$E_{^A}$}
\psfrag{EB}[]{$E_{^B}$}
\psfrag{formula01}[]{$L_{^{\half\od A,B\om\half\od A}}(t)$}
\psfrag{formula02}[]{$-\infty<t<\infty$}
 \includegraphics[width=9cm]{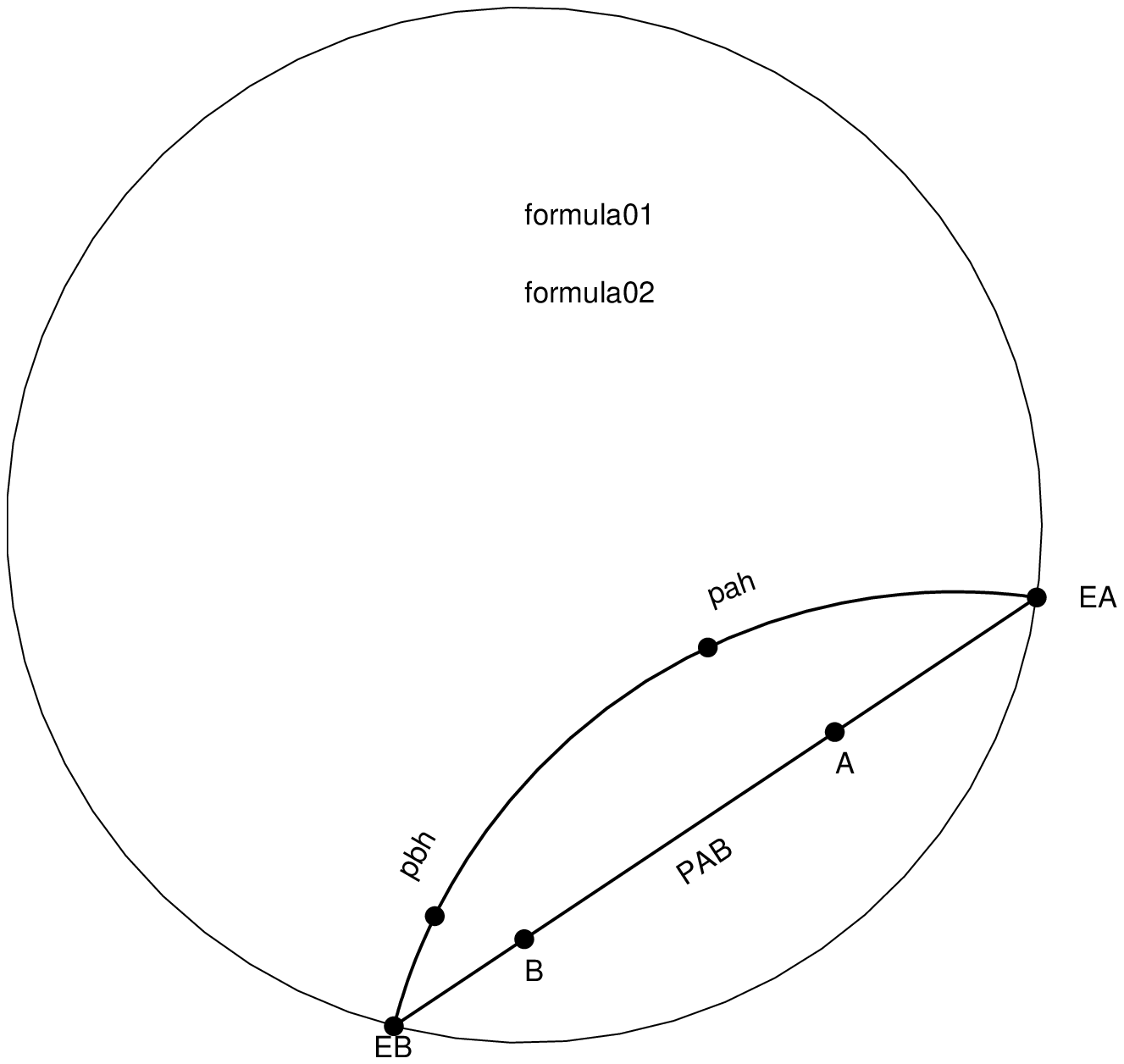}
\caption[Mobius Gyroparallelogram Addition Law]{
The Euclidean straight line (line, in short)
$\PAB(t),~-\infty<t<\infty$, \eqref{fuly02}, that passes through
the point $A$ and $B$
in a M\"obius gyrovector plane $(\Rstwo,\op,\od)$ is shown.
It is the supporting chord
of the gyroline that passes through the points $\half A$ and $B\om \half A$
in the M\"obius gyrovector plane.
The two endpoints of both the line and the gyroline,
corresponding to $t\rightarrow\pm\infty$, are $E_{^A}$ and $E_{^B}$.
\label{fig301m}}
\end{figure}


By the left loop property
of gyrations and by the gyrogroup left cancellation law \eqref{eq01b} we have
in any gyrogroup $(G,\op)$,
\begin{equation} \label{frrts}
\begin{split}
\gyr[\om A\op B, P_0] &= \gyr[\om A\op B, A] \\
&= \gyr[\om A\op B, A\op(\om A\op B)] \\
&= \gyr[\om A\op B, B] \\
&= \gyr[\om A, B]
\end{split}
\end{equation}
and
\begin{equation} \label{frrtt}
\begin{split}
\gyr[\om A\op B, P_1] &= \gyr[\om A\op B, A\op Q] \\
&= \gyr[\om A\op B, A\op (\om A\op B)] \\
&= \gyr[\om A\op B, B] \\
&= \gyr[\om A, B]
\end{split}
\end{equation}

The validity of \eqref{frut04} for all $k=0,1,2,3,\dots$
suggests the conjecture that \eqref{frut04} is valid not only for the
points of the sequence $\{P_0=A,P_1=B,P_2,P_3,\,\dots\}$
that lie on the ray $\PAB(t)$, $t\ge0$, as
shown in Fig.~\ref{fig300m}, but for all the points of the ray,
that is,
\begin{equation} \label{frut05}
\gyr[\omm A \opm B, \PAB(t)] = \gyr[\omm A ,B]
\end{equation}
for all $t\ge0$. Numerical experiments support the conjecture.

\section{Euclidean Barycentric Coordinates}
\label{dhyeuc}
\index{barycentric coordinates, Euclidean}

In order to set the stage for the introduction of hyperbolic  barycentric coordinates,
we present here the notion of Euclidean barycentric coordinates that
dates back to M\"obius' 1827 book titled {\it ``Der Barycentrische Calcul''}
(The Barycentric Calculus).
The word {\it barycenter} means
center of gravity, but the book is entirely geometrical and, hence, called
by Jeremy Gray \cite{gray93},
{\it M\"obius's Geometrical Mechanics}.
The 1827 M\"obius book is best remembered for introducing a new system
of coordinates, the {\it barycentric coordinates}.
The use of barycentric coordinates in Euclidean geometry is described in
\cite{mybook05,mybook06,yiu00}, and the
historical contribution of M\"obius' barycentric coordinates to vector analysis
is described in \cite[pp.~48--50]{crowe94}.

For any positive integer $N$,
let $m_k\in\Rb$ be $N$ given real numbers such that
\begin{equation} \label{gshmceucs}
\sum_{k=1}^{N} m_k ~\ne~0
\end{equation}
and let $A_k\in\Rn$ be $N$ given points
in the Euclidean $n$-space $\Rn$, $k=1,\ldots,N$.
Then, by obvious algebra, the equation
\begin{equation} \label{htkdbc01}
\sum_{k=1}^{N} m_k
\begin{pmatrix}  1 \\[6pt] A_k \end{pmatrix}
=
m_0
\begin{pmatrix}  1 \\[6pt] P \end{pmatrix}
\end{equation}
for the unknowns $m_0\in\Rb$ and $P\in\Rn$ possesses the unique solution
given by
\begin{equation} \label{htkdbc02}
m_0 = \sum_{k=1}^{N} m_k
\end{equation}
and
\begin{equation} \label{htkdbc03}
P = \frac{
\sum_{k=1}^{N} m_k A_k }{\sum_{k=1}^{N} m_k}
\end{equation}
satisfying for all $X\in\Rn$,
\begin{equation} \label{htkdbc04}
X+P = \frac{
\sum_{k=1}^{N} m_k (X+A_k)}{\sum_{k=1}^{N} m_k}
\end{equation}

Following M\"obius,
view \eqref{htkdbc03} as the representation of a point $P\in\Rn$
in terms of its {\it barycentric coordinates} $m_k$, $k=1,\ldots,N$,
with respect to the set of points
\begin{equation} \label{htkdbc04s}
S=\{A_1,\ldots,A_N\}
\end{equation}
Identity \eqref{htkdbc04}, then, insures that the
barycentric coordinate representation \eqref{htkdbc03} of $P$
with respect to the set $S$ is {\it covariant}\index{covariant}
(or, {\it invariant in form})
in the following sense.
The point $P$ and the points of the set $S$ of its
barycentric coordinate representation vary together under translations.
Indeed, a translation $X+A_k$ of $A_k$ by $X$, $k=1,\ldots,N$,
on the right-hand side of \eqref{htkdbc04} results in the translation
$X+P$ of $P$ by $X$ on the left-hand side of \eqref{htkdbc04}.

In order to insure that barycentric coordinate representations
with respect to a set $S$ are unique, we require the set $S$ to be pointwise independent.

\index{pointwise independence, Euclidean}
\begin{definition}\label{defptws}
{\bf (Euclidean Pointwise Independence).}
A set $S$ of $N$ points,
$S=\{A_1,\dots,A_N\}$, in $\Rn$, $n\ge2$, is {\it pointwise independent}
if the $N-1$ vectors
$-A_1+A_k$, $k=2,\dots,N$, are linearly independent in $\Rn$.
\end{definition}

We are now in the position to present
the formal definition of Euclidean barycentric coordinates, as motivated by
mass and center of momentum velocity of Newtonian particle systems.

\begin{definition}\label{defhkbde}
{\bf (Barycentric Coordinates).}
\index{barycentric coordinates, Euclidean}
Let
\begin{equation} \label{heknd}
S=\{A_1,\dots,A_N\}
\end{equation}
be a pointwise independent set of $N$ points in $\Rn$.
The real numbers $m_1,\dots,m_N$, satisfying
\begin{equation} \label{eq13ersw04}
\sum_{k=1}^{N} m_k \ne 0
\end{equation}
are barycentric coordinates of a point $P\in\Rn$ with respect to the set $S$ if
\begin{equation} \label{eq13ersw01}
P = \frac{
\sum_{k=1}^{N} m_k A_k
}{
\sum_{k=1}^{N} m_k
}
\end{equation}

Barycentric coordinates are homogeneous in the sense that
the barycentric coordinates $(m_1,\dots,m_N)$ of the point $P$
in \eqref{eq13ersw01} are equivalent to the barycentric coordinates
$(\lambda m_1,\dots,\lambda m_N)$
for any real nonzero number $\lambda\in\Rb$, $\lambda\ne0$.
Since in barycentric coordinates only ratios of coordinates are
relevant, the barycentric coordinates
$(m_1,\dots,m_N)$ are also written as $(m_1\!:\,\dots\,\!:\!m_N)$.

Barycentric coordinates that are normalized by the condition
\begin{equation} \label{eq13ersw03}
\sum_{k=1}^{N} m_k = 1
\end{equation}
are called {\it special barycentric coordinates}.\index{barycentric coordinates, special}

Equation \eqref{eq13ersw01} is said to be
the (unique) barycentric coordinate representation of $P$ with respect to the set $S$.
\end{definition}

\index{barycentric coordinates, covariance}
\begin{theorem}{~~}\label{thmfkvne}
{\bf (Covariance of Barycentric Coordinate Representations).}
Let
\begin{equation} \label{tndfkv1a}
P = \frac{
\sum_{k=1}^{N} m_k A_k
}{
\sum_{k=1}^{N} m_k
}
\end{equation}
be the barycentric coordinate representation of a point $P\in\Rn$
in a Euclidean $n$-space $\Rn$ with respect to a
pointwise independent set $S=\{ A_1,\ldots,A_N\}\subset\Rn$.
The barycentric coordinate representation \eqref{tndfkv1a}
is covariant, that is,

\begin{equation} \label{tndfkv2a}
  X + P = \frac{
\sum_{k=1}^{N} m_k (X + A_k)
}{
\sum_{k=1}^{N} m_k
}
\end{equation}
for all $X\in\Rn$, and
\begin{equation} \label{tndfkv3a}
RP = \frac{
\sum_{k=1}^{N} m_k R A_k
}{
\sum_{k=1}^{N} m_k
}
\end{equation}
for all $R\in\son$.
\end{theorem}
\begin{proof}
The proof is immediate, noting that rotations $R\in\son$ of $\Rn$ about
its origin are linear maps of $\Rn$.
\end{proof}

Following the vision of Felix Klein in his
{\it Erlangen Program} \cite{wright02}, it is owing to the covariance
with respect to translations and rotations that
barycentric coordinate representations possess geometric significance.
Indeed, translations and rotations in Euclidean geometry form the
{\it group of motions} of the geometry, and according to
Felix Klein's Erlangen Program, a geometric property is a property that
remains invariant in form under the motions of the geometry.

\section{Hyperbolic Barycentric, Gyrobarycentric, Coordinates}
\label{dhyein}
\index{barycentric coordinates, hyperbolic}
\index{gyrobarycentric coordinates}

Guided by analogies with Sec.~\ref{dhyeuc},
in this section we introduce barycentric coordinates into hyperbolic geometry
\cite{barycentric09,mybook06,mybook05}.

\index{pointwise independence, hyperbolic}
\begin{definition}\label{defptein}
{\bf (Hyperbolic Pointwise Independence).}
A set $S$ of $N$ points
$S=\{A_1,\dots,A_N\}$ in the ball $\Rsn$, $n\ge2$, is {\it pointwise independent}
if the $N-1$ gyrovectors in $\Rsn$,
$\om A_1 \op A_k$, $k=2,\dots,N$, considered as vectors in $\Rn\supset\Rsn$,
are linearly independent.
\end{definition}

We are now in the position to present
the formal definition of gyrobarycentric coordinates, that is,
hyperbolic barycentric coordinates, as motivated by
the notions of relativistic mass and
center of momentum velocity in Einstein's special relativity theory.
Gyrobarycentric coordinates, fully analogous to barycentric coordinates,
thus emerge whem Einstein's relativistic mass meets
the hyperbolic geometry of Bolyai and Lobachevsky \cite{ungarmass11}.

\begin{definition}\label{defhkbdeein}
{\bf (Gyrobarycentric Coordinates).}\index{barycentric coordinates, hyperbolic}
\index{gyrobarycentric coordinates}
Let
\begin{equation} \label{ryuhms}
S=\{A_1,\dots,A_N\}
\end{equation}
be a pointwise independent set of $N$ points in $\Rsn$.
The real numbers $m_1,\dots,m_N$, satisfying
\begin{equation} \label{eq13ersw04ein}
\sum_{k=1}^{N} m_k \gamma_{  A_k}^{\phantom{O}} > 0
\end{equation}
are gyrobarycentric coordinates of a point $P\in\Rsn$ with respect to the set $S$ if
\begin{equation} \label{eq13ersw01ein}
P = \frac{
\sum_{k=1}^{N} m_k \gamma_{  A_k}^{\phantom{O}}   A_k
}{
\sum_{k=1}^{N} m_k \gamma_{  A_k}^{\phantom{O}}
}
\end{equation}
Gyrobarycentric coordinates are homogeneous in the sense that
the gyrobarycentric coordinates $(m_1,\dots,m_N)$ of the point $P$
in \eqref{eq13ersw01ein} are equivalent to the gyrobarycentric coordinates
$(\lambda m_1,\dots,\lambda m_N)$
for any real nonzero number $\lambda\in\Rb$, $\lambda\ne0$.
Since in gyrobarycentric coordinates only ratios of coordinates are
relevant, the gyrobarycentric coordinates
$(m_1,\dots,m_N)$ are also written as $(m_1\!:\,\dots\,\!:\!m_N)$.

Gyrobarycentric coordinates that are normalized by the condition
\begin{equation} \label{eq13ersw03ein}
\sum_{k=1}^{N} m_k = 1
\end{equation}
are called {\it special gyrobarycentric coordinates}.\index{barycentric coordinates, special}

Equation \eqref{eq13ersw01ein} is said to be
the gyrobarycentric coordinate representation of $P$ with respect
to the set $S$.

Finally, the constant of the gyrobarycentric coordinate representation of $P$ in
\eqref{eq13ersw01ein} is $m_0>0$, given by
\begin{equation} \label{htkdbc02ein4}
m_0 \phantom{i} = \phantom{i} \sqrt{
\left( \sum_{k=1}^{N} m_k \right)^2 +
2\sum_{\substack{j,k=1\\j<k}}^N m_j  m_k
(\gamma_{\om  A_j \op   A_k}^{\phantom{O}} -1)
}
\end{equation}
\end{definition}

\index{gyrobarycentric coordinates, gyrocovariance}
\begin{theorem}{~~}\label{thmfkvme}
{\bf (Gyrocovariance of Gyrobarycentric Coordinate Representations).}
Let
\begin{subequations}
\begin{equation} \label{hkrfkv1a}
P = \frac{
\sum_{k=1}^{N} m_k \gamma_{  A_k}^{\phantom{O}}   A_k
}{
\sum_{k=1}^{N} m_k \gamma_{  A_k}^{\phantom{O}}
}
\end{equation}
be a gyrobarycentric coordinate representation of a point $P\in\Rsn$
in an Einstein gyrovector space $(\Rsn,\op,\od)$ with respect to a
pointwise independent set $S=\{ A_1,\ldots,A_N\}\subset\Rsn$.

Then
\begin{equation} \label{hkrfkv1b}
\gamma_{P}^{\phantom{O}} = \frac{
\sum_{k=1}^{N} m_k \gamma_{  A_k}^{\phantom{O}}
}{
m_0
}
\end{equation}
and
\begin{equation} \label{hkrfkv1c}
\gamma_{P}^{\phantom{O}} P = \frac{
\sum_{k=1}^{N} m_k \gamma_{  A_k}^{\phantom{O}} A_k
}{
m_0
}
\end{equation}
where $m_0>0$ is the constant of the
gyrobarycentric coordinate representation \eqref{hkrfkv1a} of $P$,
given by
\begin{equation} \label{hkrfkv1d}
m_0 \phantom{i} = \phantom{i} \sqrt{
\left( \sum_{k=1}^{N} m_k \right)^2 +
2\sum_{\substack{j,k=1\\j<k}}^N m_j  m_k
(\gamma_{\om  A_j \op   A_k}^{\phantom{O}} -1)
}
\end{equation}

Furthermore, the gyrobarycentric coordinate representation \eqref{hkrfkv1a}
and its associated identities in \eqref{hkrfkv1b}\,--\,\eqref{hkrfkv1d}
are gyrocovariant, that is,

\end{subequations}
\begin{subequations}
\begin{equation} \label{hkrfkv2a}
  X \op  P = \frac{
\sum_{k=1}^{N} m_k \gamma_{  X \op   A_k}^{\phantom{O}} (  X \op   A_k)
}{
\sum_{k=1}^{N} m_k \gamma_{  X \op   A_k}^{\phantom{O}}
}
\end{equation}
\begin{equation} \label{hkrfkv2b}
\gamma_{ X  \op   P}^{\phantom{O}} = \frac{
\sum_{k=1}^{N} m_k \gamma_{ X  \op   A_k}^{\phantom{O}}
}{
m_0
}
\end{equation}
\begin{equation} \label{hkrfkv2c}
\gamma_{ X  \op P}^{\phantom{O}} ( X  \op P)= \frac{
\sum_{k=1}^{N} m_k \gamma_{ X  \op   A_k}^{\phantom{O}}
( X  \op   A_k)
}{
m_0
}
\end{equation}
\begin{equation} \label{hkrfkv2d}
m_0 \phantom{i} = \phantom{i} \sqrt{
\left( \sum_{k=1}^{N} m_k \right)^2 +
2\sum_{\substack{j,k=1\\j<k}}^N m_j  m_k
(\gamma_{\om( X \op  A_j)\op( X \op  A_k)}^{\phantom{O}} -1)
}
\end{equation}
for all $X\in\Rsn$, and
\end{subequations}
\begin{subequations}
\begin{equation} \label{hkrfkv3a}
RP = \frac{
\sum_{k=1}^{N} m_k \gamma_{R A_k}^{\phantom{O}} R A_k
}{
\sum_{k=1}^{N} m_k \gamma_{R A_k}^{\phantom{O}}
}
\end{equation}
\begin{equation} \label{hkrfkv3b}
\gamma_{RP}^{\phantom{O}} = \frac{
\sum_{k=1}^{N} m_k \gamma_{R A_k}^{\phantom{O}}
}{
m_0
}
\end{equation}
\begin{equation} \label{hkrfkv3c}
\gamma_{RP}^{\phantom{O}} (RP)= \frac{
\sum_{k=1}^{N} m_k \gamma_{R A_k}^{\phantom{O}}
(R A_k)
}{
m_0
}
\end{equation}
\begin{equation} \label{hkrfkv3d}
m_0 \phantom{i} = \phantom{i} \sqrt{
\left( \sum_{k=1}^{N} m_k \right)^2 +
2\sum_{\substack{j,k=1\\j<k}}^N m_j  m_k
(\gamma_{\om(R A_j)\op(R A_k)}^{\phantom{O}} -1)
}
\end{equation}
for all $R\in\son$.
\end{subequations}
\end{theorem}

The proof of Theorem \ref{thmfkvme} is presented in
\cite[Theorem 4.4]{mybook06} and \cite[Theorem 4.6]{mybook05}.

\begin{remark}\label{remfhk}
It is assumed in Theorem \ref{thmfkvme} that the point
$P$ in \eqref{eq13ersw01ein}
lies inside the ball $\Rsn$, implying that $m_0^2>0$ and that the gamma factor
$\gamma_P$ of $P$ is a real number.
The constant $m_0$ of a gyrobarycentric coordinate representation
\eqref{eq13ersw01ein} of a point $P$ determines whether $P$ lies
inside the ball $\Rsn$.

If the coefficients $m_k$, $k=1,\ldots,N$, in the
gyrobarycentric coordinate representation \eqref{eq13ersw01ein} of $P$
are all positive or all negative,
then the point $P$ lies in the convex span\index{convex span}
of the points of the set $S$, that is, $P$ lies inside the $(N-1)$-gyrosimplex
$A_1\ldots A_N$. This gyrosimplex, in turn, lies inside the ball $\Rsn$.
\begin{itemize}
\item[$(1)$]
The point $P$ lies inside the $(N-1)$-gyrosimplex $A_1\ldots A_N$ if and only if
the coefficients $m_k$, $k=1,\ldots,N$, of its gyrobarycentric coordinate representation
\eqref{hkrfkv1a} are all positive or all negative. Clearly, in this case $m_0^2>0$.
\end{itemize}

Otherwise, when all the coefficients $m_k$ are nonzero but do not have
the same sign, the location of $P$ has the following three possibilities
that correspond to whether the gamma factor \eqref{hkrfkv1b} of $P$
is real, infinity, or imaginary:
\begin{itemize}
\item[$(2)$]
The point $P$ does not lie inside the $(N-1)$-gyrosimplex $A_1\ldots A_N$, but it
lies inside the ball $\Rsn$. In this case the gamma factor
$\gamma_P$ of $P$ is a real number and, hence, $m_0^2>0$.
\item[$(3)$]
The point $P$ lies on the boundary of the ball $\Rsn$ if and only if the gamma factor
$\gamma_P$ of $P$ is undefined, $\gamma_P=\infty$, so that $m_0^2=0$.
\item[$(4)$]
The point $P\inn\Rn$ does not lie in the ball $\Rsn$ or on its boundary
if and only if the gamma factor $\gamma_P$ of $P$ is purely imaginary,
so that $m_0^2<0$.
\end{itemize}
\end{remark}

Examples for the use of gyrobarycentric coordinates for the
determination of several hyperbolic triangle centers are found
in \cite{mybook06,mybook05}.

Employing the technique of gyrobarycentric coordinate representations,
we will now determine the end points
$E_{^A}$ and $E_{^B}$ of a gyroline $\LAB(t)$ in an
Einstein gyrovector space $(\Rsn,\op,\od)$,
shown in Fig.~\ref{fig174b3enm}, p.~\pageref{fig174b3enm}.

Let $A_1,A_2\in\Rsn$ be two distinct points of
an Einstein gyrovector space $(\Rcn,\op,\od)$,
and let $P$ be a generic point on the gyroline, \eqref{frut01},
\begin{equation} \label{grut00}
P_{12}(t) = A_1\op(\ome A_1\op A_2)\od t
\end{equation}
$t\in\Rb$, that passes through these two points. Furthermore, let
\begin{equation} \label{grut01}
P = \frac{
m_1 \gamma_{  A_1}^{\phantom{O}} A_1 + m_2 \gamma_{  A_2}^{\phantom{O}} A_2
}{
m_1 \gamma_{  A_1}^{\phantom{O}} + m_2 \gamma_{  A_2}^{\phantom{O}}
}
\end{equation}
be the gyrobarycentric coordinate representation of $P$ with respect
to the pointwise independent set $S=\{A_1,A_2\}$, where the
gyrobarycentric coordinates $m_1$ and $m_2$ are to be determined.
Owing to the homogeneity of gyrobarycentric coordinates, we can select $m_2=-1$,
obtaining from \eqref{grut01} the gyrobarycentric coordinate representation
\begin{equation} \label{grut02}
P = \frac{
m   \gamma_{  A_1}^{\phantom{O}} A_1 -\gamma_{  A_2}^{\phantom{O}} A_2
}{
m   \gamma_{  A_1}^{\phantom{O}} -\gamma_{  A_2}^{\phantom{O}}
}
\end{equation}

According to Def.~\ref{defptein} of the gyrobarycentric coordinate representation
of $P$ in \eqref{eq13ersw01ein} and its constant $m_0$ in \eqref{htkdbc02ein4},
the constant $m_0$ of the gyrobarycentric coordinate representation of $P$
satisfies the equation
\begin{equation} \label{grut03}
\begin{split}
m_0^2 &= m_1^2 + m_2^2 + 2m_1m_2 \gamma_{\om  A_1 \op A_2}^{\phantom{O}}
\\&= m^2+1+2m\gammaab
\end{split}
\end{equation}
where we use the convenient notation
\begin{equation} \label{htrid}
\begin{split}
\ab_{ij} &= \om A_i\op A_j \\
\gamma_{ij}^{\phantom{O}} &= \gamma_{\ab_{ij}}^{\phantom{O}}
\end{split}
\end{equation}
$i,j\in\NN$.

As remarked in Item (3) of Remark \ref{remfhk}, the point $P$ lies on the
boundary of the ball $\Rsn$ if and only if $m_0=0$, that is by \eqref{grut03},
if and only if
\begin{equation} \label{grut04}
m^2-2m\gammaab+1=0
\end{equation}
The two solutions of \eqref{grut04} are
\begin{equation} \label{grut05}
\begin{split}
m &= \gammaab + \sqrt{\gamma_{12}^2-1}
\\[8pt]
m &= \gammaab - \sqrt{\gamma_{12}^2-1}
\end{split}
\end{equation}

The substitution into \eqref{grut02} of each of the two solutions \eqref{grut05}
gives the two endpoints $E_{^{A_1}}$ and $E_{^{A_2}}$ of the
gyroline $P_{12}(t)$ in \eqref{grut00},
\begin{equation} \label{grut06}
\begin{split}
E_{^{A_1}} &= \frac{
(\gammaab+\sqrt{\gamma_{12}^2-1}) \gamma_{A_1}^{\phantom{O}}A_1 - \gamma_{A_2}^{\phantom{O}}A_2
}{
(\gammaab+\sqrt{\gamma_{12}^2-1}) \gamma_{A_1}^{\phantom{O}}    - \gamma_{A_2}^{\phantom{O}}
}
\\[8pt]
E_{^{A_2}} &= \frac{
(\gammaab-\sqrt{\gamma_{12}^2-1}) \gamma_{A_1}^{\phantom{O}}A_1 - \gamma_{A_2}^{\phantom{O}}A_2
}{
(\gammaab-\sqrt{\gamma_{12}^2-1}) \gamma_{A_1}^{\phantom{O}}    - \gamma_{A_2}^{\phantom{O}}
}
\end{split}
\end{equation}
which are shown in Fig.~\ref{fig174b3enm}, p.~\pageref{fig174b3enm}, for
$A_1 = A$ and $A_2 = B$.

The expressions for $\om A_1\op E_{^{A_1}}$ and $\om A_1\op E_{^{A_2}}$
that follow from \eqref{grut06} by means of the
gyrocovariance identity \eqref{hkrfkv2a} in Theorem \ref{thmfkvme} are particularly
elegant.
Indeed, by the gyrocovariance identity \eqref{hkrfkv2a} with $X=\om A_1$, applied to
each of the two equations in \eqref{grut06}, we have
\begin{equation} \label{gurimd}
\begin{split}
\om A_1 \op E_{^{A_1}} &= \frac{
(\gammaab + \sqrt{\gamma_{12}^2-1})
\gamma_{\om A_1 \op A_1}^{\phantom{O}} (\om A_1 \op A_1)
-
\gamma_{\om A_1 \op A_2}^{\phantom{O}} (\om A_1 \op A_2)
}{
(\gammaab + \sqrt{\gamma_{12}^2-1})
\gamma_{\om A_1 \op A_1}^{\phantom{O}}
-
\gamma_{\om A_1 \op A_2}^{\phantom{O}}
}
\\[4pt]
&= \frac{-\gammaab\ab_{12}}{(\gammaab+\sqrt{\gamma_{12}^2-1})-\gammaab}
\\[4pt]
&= \om\frac{\gammaab\ab_{12}}{\sqrt{\gamma_{12}^2-1}}
\\[8pt]
\om A_1 \op E_{^{A_2}}
&= \phantom{\om}\frac{\gammaab\ab_{12}}{\sqrt{\gamma_{12}^2-1}}
\end{split}
\end{equation}
where we use the notation \eqref{htrid}, noting that
$\om A_1 \op A_1=\zerb$ and
$\gamma_{\om A_1 \op A_1}^{\phantom{O}}=\gamma_{\zerb}^{\phantom{O}}=1$.

The equations in \eqref{gurimd} imply, by means of the
left cancellation law \eqref{eq01b},
\begin{equation} \label{gurime}
\begin{split}
 E_{^{A_1}} &= A_1 \om\, \frac{\gammaab\ab_{12}}{\sqrt{\gamma_{12}^2-1}}
\\
 E_{^{A_2}} &= A_1 \op\, \frac{\gammaab\ab_{12}}{\sqrt{\gamma_{12}^2-1}}
\end{split}
\end{equation}

Interestingly, \eqref{gurime} remains invariant in form under
the isomorphism \eqref{eq3ejksa}, as seen from \eqref{gtres}.
Accordingly, the equations in \eqref{gurime} with $\op=\ope$ being Einstein addition are
used in calculating the endpoints of an Einstein gyroline
in Fig.~\ref{fig174b3enm}, p.~\pageref{fig174b3enm}, and
the same equations \eqref{gurime}, but with $\op=\opm$ being M\"obius addition are
used in calculating the endpoints of a M\"obius gyroline
in Fig.~\ref{fig174a1m}, p.~\pageref{fig174a1m}.

In Fig.~\ref{fig303m}, p.~\pageref{fig303m}, the points
$A,B\in(\Rctwo,\opm,\od)$ of a M\"obius gyrovector plane
are shown along with their respective isomorphic images
$2\od A,2\od B\in(\Rctwo,\ope,\od)$ of an Einstein gyrovector plane,
under the isomorphism \eqref{hdkeb1}.
Indeed, as expected, Fig.~\ref{fig303m} indicates that the
endpoints $E_{^A}$ and $E_{^B}$ of
\begin{enumerate}
\item \label{fvlb1}
the M\"obius gyroline (a circular arc) through the points $A$ and $B$, and of
\item \label{fvlb2}
the Einstein gyroline (a chord) through the points $2\od A$ and $2\od B$,
\end{enumerate}
are coincident.

This Chapter appears in \cite{ungar12s}.

\end{document}